\documentclass[
reprint,
pra,
superscriptaddress,
floatfix,
citeautoscript
]{revtex4-2}

\usepackage[dvipsnames]{xcolor}
\usepackage{amsthm}
\usepackage{thmtools}
\usepackage{mathtools}
\usepackage{physics}
\usepackage{bm} 
\usepackage{bbm} 
\usepackage{braket}
\usepackage{dsfont}
\usepackage{amsmath}
\usepackage{amssymb}
\usepackage{enumitem}
\usepackage{caption}
\usepackage{ragged2e}
\usepackage{subcaption}
\usepackage{tikz}
\usepackage{qcircuit}
\usetikzlibrary{arrows, positioning}
\usepackage{tikz-cd}
\usepackage[a4paper, total={6.8in, 9.5in}]{geometry}

\usepackage{chngcntr}
\usepackage{apptools}
\AtAppendix{\counterwithin{lemma}{section}}
\AtAppendix{\counterwithin{proposition}{section}}
\AtAppendix{\counterwithin{definition}{section}}
\AtAppendix{\counterwithin{corollary}{section}}
\AtAppendix{\counterwithin{subsection}{section}}
\AtAppendix{\counterwithin{equation}{section}}
\AtAppendix{\counterwithin{figure}{section}}

\usepackage{algorithm}
\usepackage{algpseudocode}

\usepackage{comment}

\usepackage[colorlinks]{hyperref}
\hypersetup{
    colorlinks  = true,
    citecolor   = PineGreen,
    linkcolor   = MidnightBlue,
    urlcolor    = PineGreen
}

\declaretheorem[style=plain]{theorem}

\declaretheorem[style=plain,sibling=theorem]{corollary}
\declaretheorem[style=plain,sibling=theorem]{lemma}
\declaretheorem[style=plain,sibling=theorem]{proposition}
\declaretheorem[style=definition,sibling=theorem]{definition}

    \DeclareMathOperator{\Ima}{Im} 
    
\begin{document}

\title{Quantum Walks: First Hitting Times with Weak Measurements}

\author{Tim Heine}
\email{tim.heine@dlr.de}
\affiliation{Institute of Quantum Technologies, German Aerospace Center (DLR), Wilhelm-Runge-Str. 10, D-89081 Ulm, Germany}

\author{Eli Barkai}
\affiliation{Department of Physics, Institute of Nanotechnology and Advanced Materials, Bar-Ilan University, Ramat-Gan, 52900, Israel}

\author{Klaus Ziegler}
\affiliation{Institut für Physik, Universität Augsburg, D-86135 Augsburg, Germany}

\author{Sabine Tornow}
\affiliation{Research Institute CODE, Universität der Bundeswehr München, D-81739 Munich, Germany}

\date{\today}

\begin{abstract}
We study the first detected recurrence time problem of continuous-time quantum walks on graphs. While previous works have employed projective measurements to determine the first return time, we implement a protocol based on weak measurements on a dilated system, enabling minimally invasive monitoring throughout the evolution.
To achieve this, we extend the theoretical framework and complement it with both numerical simulations and experimental investigations on an IBM quantum computer. Despite the implementation of a generalized measurement, our modified formalism of weak recurrence provides a description purely within the Hilbert space of the quantum system. Our results reveal that the first hitting time scales inversely with the coupling parameter between the ancilla and the quantum system.
\end{abstract}

\maketitle

\section{\label{section:Introduction}Introduction}
%
%
Quantum random walks are universal for quantum computation with demonstrated exponential speed-up over random walks \cite{childs2003exponential,kempe2003quantum, kempe2005discrete}. Being the quantum counterpart of classical random walks, quantum walks find applications in stochastic processes \cite{whitfield2010quantum}, quantum search algorithms \cite{shenvi2003quantum}, quantum computation \cite{childs2009universal}, and many more \cite{qiang2024review}.  
%
%
%
\textit{Quantum recurrence} indicates the ability of a quantum walker to return to its initial location \cite{nitsche2018probing} with probability one.  
%
%
%
%
%
%
The \textit{first hitting time} of a quantum walk characterizes the expected time, which a quantum particle takes to evolve through a lattice and \textit{hit} its initial site for the first time. The phrases \textit{hitting time, return time,} and \textit{recurrence time} are used synonymously. Classical first hitting times on graphs \cite{grebenkov2024target, redner2001guide} are utilized for instance in cellular biology \cite{iyer2016first}, and their quantum counterpart finds applications in quantum chemistry \cite{kim1958mean} and quantum networks \cite{chiribella2009theoretical}.
%
%
%
%
%

To detect the position of a walker obeying the laws of quantum theory, one has to measure the system. Axiomatically, a quantum mechanical measurement intercepts the evolution and hence \textit{changes} the quantum state. Thus, a \textit{measurement-indu\-ced} quantum walk refers to a protocol, where a quantum system evolves unitarily for some time and is repeatedly measured. This repeated measurement procedure is called \textit{monitoring} and there is no unique way to specify such a protocol \cite{varbanov2008hitting}.
The resulting dynamics depends sensitively on the type of monitoring. 
There are two canonical choices for choosing the \textit{rate} at which the system is measured. 
First, one can measure stroboscopically at a constant rate. This means that the system is \textit{always} measured at constant frequency.  
Within this protocol, any quantum particle in a finite-dimensional system will \textit{definitely} return to its initial state after a finite time period \cite{grunbaum2013recurrence}. Hence, the first hitting time provides information about size, complexity and connectivity of a quantum system.  
Second, after a fixed time interval, loosely speaking, one \textit{flips a coin,} whose outcome decides whether the quantum system is measured or being let free to evolve \cite{varbanov2008hitting, kessler2021first}.
We will focus on the first protocol, i.e. we repeat the measurement of a tight binding continuous-time quantum walk at a fixed, constant rate. 

In quantum computer science, an elementary algorithm consists of two basic components: unitary evolution and measurement. However, there a two conceptual approaches in designing quantum circuits: \textit{Static} circuits, where the measurement operation is executed at the \textit{end} of a unitary sequence and \textit{dynamic} circuits \cite{corcoles2021exploiting, baeumer2024quantum}, which employ measurements within the \textit{middle} of the sequence. In light of quantum circuit design, the monitored quantum walk represents a minimal version of a dynamic circuit.
Quantum mechanically, a \textit{strong} measurement refers to an idealized, projection-valued measurement, with no further disturbance to the system.
Regarding a quantum algorithmic framework, where the unitary operator acts on a (finite) register of qubits, a measurement is often realized by additional \textit{ancilla} qubits that are utilized for \textit{weak} measurements. 
In practice, these measurements affect the quantum dynamics via less disturbance than strong, projection-valued measurements, which classifies these measurements as \textit{weak.} Weak measurements indicate that the interaction between the system and the measurement apparatus is sufficiently small, resulting in limited information gain per measurement and minimal disturbance to the system. This may manifest as \textit{increased uncertainty} in the measurement outcomes, which can resemble additional randomness.  
This increased uncertainty is induced by the parameterized coupling strength between the two systems, which affects the outcome of the measurement.
Thus, the (weak) first hitting time indicates an \textit{expected number of (weak) measurements} which it takes for the particle to be detected for the first time. This must not to be confused with a physical time. 
Hence, the event of first detection is a sequence of binary outcomes, dependent on whether the quantum walker (i.e. a quantum particle) has been detected "yes", or not "no" \cite{yin2023restart}. In case that the detector clicks ("yes", click measurement), the protocol is terminated. If the detector does not click ("no", no-click measurement), the protocol is continued while for weak measurements the ancilla qubit influences the system qubit \cite{Dubey_2023}.

First hitting times for the continuous-time quantum walk with generalized weak measurements were already studied by Varbanov et al. \cite{varbanov2008hitting}. Their protocol differs from our weak measurement protocol in the way that we always measure (stroboscopic monitoring), once the unitary evolution is stopped. Their protocol includes a third event which corresponds to letting the quantum system evolve freely, that is, no measurement at all. 
The fundamental theoretical framework, which is extensively utilized in this work, was developed by Grünbaum et al. \cite{grunbaum2013recurrence}.
There, the authors derived a topologically protected, spectral characterization of quantum recurrence for discrete-time unitaries based on properties of Schur functions \cite{schur1917potenzreihen}, which was extended by Bourgain et al.  \cite{bourgain2014quantum} and generalized to open quantum walks \cite{grunbaum2018generalization, grunbaum2020quantum}. 
Experimentally, first hitting times were already evaluated on a quantum computer through a strong, projective measurement protocol in \cite{measurement2023tornow, wang2024first, yin2025restart}. 
%
%
%
%

In this work, we compute the expected first return time of a weakly monitored continuous-time quantum walk thr\-ough graphs by applying weak measurements stroboscopically.
The\-se measurements are implemented by a controlled rotation through one additional ancilla qubit around the Pauli-Y axis. The return time is computed in dependence of the coupling strength between the ancilla qubit and the quantum system. Although such a protocol is rigorously described by generalized operator-valued measurements, the specific choice of the rotation enables a description within vectors in Hilbert spaces.
Hence, we can map such a protocol to the Hilbert space framework of quantum recurrence \cite{grunbaum2013recurrence, bourgain2014quantum}, and analyze the effect of the coupling strength. 
Therefore, our description avoids the system's dilation and builds upon fundamental theory developed in \cite{grunbaum2013recurrence,bourgain2014quantum,grunbaum2018generalization,grunbaum2020quantum,cedzich2016bulk,cedzich2018topological,cedzich2019quantum}.
Moreover, we experimentally demonstrate our modified measurement-induced first hitting time on an IBM quantum computer underscoring the recurrence theory's resilience to weak couplings. 

The structure of this paper is outlined as follows.
At first, we recapture monitored quantum walks through projection-valued measurements in Section \ref{section:Monitored_Quantum_Walks}. Secondly, we propose our weak monitoring protocol in Section \ref{section:Indirect_Measurements}. Thirdly, the weakly monitored first hitting time is analyzed in Section \ref{section:Indirect_First_Hitting_Time}, and finally, we propose further research directions in Section \ref{section:Conclusions}.

\section{\label{section:Monitored_Quantum_Walks}Monitored Quantum Walks}
The stochastic process under investigation is a contin\-uous-time quantum walk of a single quantum state evolving through a finite graph.
The unitary evolution is generated by a time independent Hamilton operator which is given by the adjacency matrix of the graph.
Let $H$ denote that adjacency matrix. The corresponding quantum walker evolves unitarily according to the Schrödinger equation via
\begin{equation}\label{eq:continuous_time_unitary}
    U(t)=e^{-iH t},
\end{equation}
such that the state of the system at time $t$ is given by $\ket{\psi(t)} = U(t) \ket{\psi(0)}$, for some known initial state $\ket{\psi(0)}.$ The \textit{continuity} of the map $t\mapsto e^{-iHt}$ gives the \textit{continuous-time quantum walk} its name.
If the unitary evolution is interspersed with measurements, the resulting quantum walk is called \textit{monitored quantum walk (MQW)} \cite{varbanov2008hitting}. If the unitary is chosen to be of the form (\ref{eq:continuous_time_unitary}), one calls the resulting protocol a \textit{monitored continuous-time quantum walk (MCTQW).}
\subsection{State Recurrence and Subspace Recurrence}
We consider the first detected return time problem \cite{wang2024first, friedman2017quantum}. Contrary to classical random walks, where the observer can detect the position of the walker without interfering with the system's dynamics, a quantum walk must be subjected to repeated measurements.
We will monitor the quantum system stroboscopically for $N$ times at a constant rate. The real-valued time interval, at which the quantum system evolves unitarily in between two consecutive measurements is denoted by $t$. The resulting stochastic process is a Markov chain \cite{grunbaum2013recurrence}. For example, the event that the particle has been firstly detected after the third measurement, given that it was not detected twice beforehand, is written "no-no-yes" or simply "001".

Throughout, we will follow the notation by Bourgain et al. \cite{bourgain2014quantum}, and define for an initial quantum state $\ket{\psi}$, a subspace $V\subset \mathcal{H}$, which will be repeatedly measured. The measurement operator, which describes a projection onto that subspace, is denoted by $P$. Hence, for $\ket{\psi}\in V$,  $V:=\Ima P$ denotes the image of the projection operator, which is a closed subspace in $\mathcal{H}.$ The set of quantum states on $V$ is defined as \cite{bourgain2014quantum}
\begin{equation*}
    S_V := \left\{ \ket{\psi}\in V : \norm{\ket{\psi}} = 1 \right\}.
\end{equation*}
Let $Q:=\mathds{1}-P$ be the complementary projector onto $V^{\perp}:= \Ima Q.$
For shorter notation, we will sometimes denote $U(t)$ as $U_t$.
The operators $PU_t$ and $QU_t$ are referred as \textit{termination} and \textit{survival operator.} This means that a successful measurement (i.e. with $P$) \textit{terminates} the protocol, whereas an unsuccessful measurement lets the walker \textit{survive.} 
To ensure that an initial state $\ket{\psi}$ is contained in $V$, Bourgain et al. \cite{bourgain2014quantum} defined the linear \textit{first return amplitude operator} $\hat{a}_n(t)$ to be
\begin{equation}\label{eq:first_return_amplitude_op}
    \Hat{a}_n(t):= PU_t \tilde{U_t}^{n-1}P, \qquad \Tilde{U_t}:= Q U_t
\end{equation}
\subsubsection{First Hitting Time}
For a stochastic process with initial state $\ket{\psi}\in S_V$, the \textit{detection probability} $p_n\equiv p_n(t)$ is referred to the probability to detect the quantum walker after $n\in\mathbb{N}$ steps \cite{bourgain2014quantum}
\begin{equation}\label{return_probability}
    p_n(t):= \norm{\Hat{a}_n(t)\ket{\psi}}^2.
\end{equation}
Then, the overall return probability $R_N$ is defined as the total probability that the process \textit{returns} to its initial state after $N$ steps, i.e. 
\begin{equation*}
    R_N:= \sum_{n=1}^N p_n, \qquad R_{\infty}:= \lim_{N\to\infty} R_N. 
\end{equation*}
Since the experiment on a quantum computer requires a protocol with \textit{finite} circuit depth $N$, we study any involved stochastic quantity also with regard to a finite sample space. 
Note that $R_N$ technically depends on the time $t$ and the choice of the initial state $\ket{{\psi}}\in S_V$.  Here, $\norm{\cdot} \equiv \sqrt{ \langle \cdot, \cdot \rangle }$ denotes the usual Hilbert space norm. 
The \textit{expected return time} $\tau_N$ refers to the \textit{expected number of measurements,} which it takes for the quantum walker to return to its initial subspace. The quantity $\tau_N$ becomes a \textit{physical} time, if we multiply it with $t$ so that $t_N: = \tau_N t$ represents the \textit{physical expected first return time.} $t_N$ and $\tau_N$ are proportional by $t$ so that we restrict our analysis to $\tau_N$.

For an initial state $\ket{\psi}\in S_V$, $\tau_N \equiv \tau_N(t)$ is the \textit{expectation value} of the first return event, which is given by \cite{bourgain2014quantum}
\begin{equation}\label{return_time}
    \tau_N(t):=  \frac{\sum_{n=1}^N n \, p_n(t)}{R_N(t)}.
\end{equation}
Note that
$ \tau_{\infty}:= \lim_{N\to\infty} \tau_{N} = \lim_{N\to\infty}\sum_{n=1}^N n p_n, $
because here $\lim_{N\to\infty} R_N = 1.$
Under the assumption that $\mathcal{H}$ is finite-dimensional, any closed subspace $V$ is automatically V-recurrent \cite{bourgain2014quantum, grunbaum2018generalization}. 
\subsection{Generalized Quantum Measurements}
In Sec. \ref{section:Indirect_Measurements}, we will consider the first hitting time with weak measurements. Here, we repeat the basic background for operator-valued measurements. 

Generalizing the quantum measurement from projection-valued measurements (PVMs) on a Hilbert space to positive operator-valued measurements (POVMs) on a dilated system is the unique mathematical construction of modeling imperfect quantum measurements \cite{von2018mathematical,nielsen2010quantum}. Those quantum walks, which are measured via an interaction with an ancilla system (as we consider it in this paper) are classified as \textit{open quantum walks} \cite{grunbaum2018generalization}. 
General mixed quantum states are described by density operators, i.e. Hermitian, positive semi-definite operators with trace one. 
A sub-normalized density operator is a density operator with trace less or equal to one. 
In quantum information, any general CPTP channel, i.e. a \textit{completely positive} (CP) and \textit{trace preserving} (TP) linear map is of the form 
\begin{equation}\label{eq:quantum_channel}
  \mathcal{E}(\rho)=\sum_{i=1}^n K_i \rho K_i^{\dagger},
\end{equation}
with $K_i$ being \textit{Kraus operators} and  the effect operators $E_i := K_i K_i^{\dagger}$ are positive with $0 \leq E_i \leq \mathds{1}$ and satisfy the completeness relation $\sum_{i=1}^n E_i = \mathds{1}.$
The first return amplitude operator from (\ref{eq:first_return_amplitude_op}) becomes a subnormalized \textit{first return channel}, which is \cite{grunbaum2018generalization}
\begin{equation}\label{eq:first_return_channel}
    A_n (\rho):= \Tilde{a}_n\, \rho\, \Tilde{a}_n^{\dagger},\quad \forall\, n\geq 1.
\end{equation}
with $\Tilde{a}_n \equiv \Tilde{a}_n(t)$ similar to (\ref{eq:first_return_amplitude_op}), but now modified to a generalized measurement, which reads
\begin{equation}\label{eq:first_return_povm}
    \Tilde{a}_n := K_{\text{yes}} U \Tilde{U}^{n-1}, \qquad \Tilde{U}:= K_{\text{no}} U,\qquad n \geq 1. 
\end{equation}
The probability of detecting the walker at the $n$th measurement, given that the particle is not detected $n-1$ times, is hence
\begin{equation}
    p_n = \Tr A_n(\rho).
\end{equation}
The Kraus operators $K_{\text{yes}}$ and $K_{\text{no}}$ indicate a binary measurement, which uniquely refers to the events of having detected the particle "yes" or not "no". The evolution in (\ref{eq:first_return_channel}) was already proposed by Krovi et al. \cite{krovi2006quantum}. 
Stinespring's dilation theorem states that any POVM is a PVM on a larger, \textit{dilated} space \cite{stinespring1955positive}. However, the converse does not hold true. In Section \ref{section:Indirect_Measurements}, we exactly proceed along such a converse direction.

%

The unitary evolution of a density operator $\rho\mapsto U\rho U^{\dagger}$ can be represented as a matrix-vector-multi\-plica\-tion on a space of squared dimension \cite{horn1994topics,grunbaum2020quantum}, as
\begin{equation}\label{eq:vectorization}
    vec \left( U \rho U^{\dagger} \right) = \left( U \otimes \overline{U} \right)\,vec \,(\rho).
\end{equation}
For a complex-valued, square matrix $A = (a_{ij})_{i,j=1}^n$, the bijective \textit{vectorization} mapping $vec (A): M_n(\mathbb{C})\to \mathbb{C}^{n^2}$ is $vec (a_{ij}) = a_{(i-1)n+j}$.
Therefore, one can use the spectral recurrence formalism \cite{grunbaum2013recurrence, bourgain2014quantum} subjected to operator-valued measurements of any quantum channel as discussed in \cite{grunbaum2018generalization, grunbaum2020quantum}.

\section{\label{section:Indirect_Measurements}Monitoring via Weak Measurements}
\subsection{Weak Measurements}
To the purpose of simplicity, we will describe first hitting times through weak measurements solely by using subnormalized pure states and projection-valued measurements. By doing so, we can utilize the spectral characterization of recurrence developed in \cite{grunbaum2013recurrence} for our approach of measuring weakly.  

If we are able to track the action of the POVM back to a modified \textit{pseudo-PVM} on the same (Hilbert) space of the quantum system (and not on the dilated system), we are able to describe the monitored protocol nevertheless through vectors on Hilbert spaces. For quantum recurrence, there already exists a comprehensive theory of both (scalar-valued) state recurrence \cite{grunbaum2013recurrence} and (operator-valued) subspace recurrence \cite{bourgain2014quantum, grunbaum2018generalization}. 
As we will see, our choice of weak measurements can be mathematically embedded into the framework of quantum state recurrence \cite{grunbaum2013recurrence}. 

Our choice of a weak measurement is operationally equivalent to an evolution through a \textit{quantum phase damping channel} \cite{nielsen2010quantum}. This channel shears the projection operators in dependence on the rotation angle between the ancilla and the quantum system.
Physically speaking, such a channel describes the event (ocurring with probability $\eta$) that a photon has been scattered without loss of energy \cite{nielsen2010quantum}. The corresponding Kraus operators are
\begin{equation}\label{eq:Kraus_operators}
     K_{\text{yes}}(\eta) = \begin{pmatrix}
        0 & 0 \\
        0 & \sqrt{\eta} \\
    \end{pmatrix},
    \quad 
     K_{\text{no}}(\eta) = \begin{pmatrix}
        1 & 0 \\
        0 & \sqrt{1-\eta} \\
    \end{pmatrix},
\end{equation}
It is emphasized that these operators satisfy the subsequent \textit{pseudo-completeness relation}
\begin{align}\label{eq:pseudo_complete_kraus}
    &K_{\text{yes}}(\eta) + K_{\text{no}}(\eta) = \mathds{1} + f(\eta) \begin{pmatrix}
        0 & 0 \\
        0 & 1 \\
    \end{pmatrix},
\end{align}
\begin{align}\label{eq:f_eta}
    \qquad &f(\eta) :=\sqrt{\eta}+\sqrt{1-\eta}-1. 
\end{align}
With \textit{pseudo-completeness} we mean that $K_{\text{yes}}$ and $K_{\text{no}}$ resolve the identity plus an $\eta-$ dependent perturbation by some diagonal operator. Here, that operator is just $\ketbra{1}{1}\equiv \begin{pmatrix}
        0 & 0 \\
        0 & 1 \\
    \end{pmatrix}$.
Relation (\ref{eq:pseudo_complete_kraus}) is essential for our protocol: in what follows, we map the operator-valued weak measurements (which are realized by a controlled $R_Y$ rotation on the ancilla system) to the projection-valued recurrence theory through rank-one projections developed in \cite{grunbaum2013recurrence} and extended in \cite{bourgain2014quantum}. Our approach is built upon (\ref{eq:pseudo_complete_kraus}), as it will become clear in Section \ref{subsubseq:weak_first_return_amplitude_operator}. However, any POVM measurement, which does not fulfill such a pseudo-completeness relation cannot be canonically embedded (without dilation) into a recurrence formalism for pure states. In general, the wider theoretical framework of \textit{open quantum walks} \cite{grunbaum2018generalization, grunbaum2020quantum} must be taken into account.
We consider the Hilbert space $\mathcal{H}$ of two composite systems, i.e. $\mathcal{H} = \mathcal{H}_{\text{sys}}\otimes \mathcal{H}_{\text{anc}}$. The first subsystem represents a register of $d$ qubits $\mathcal{H}_{\text{sys}}=\mathbb{C}^{2^d}$, which underlies a Hamiltonian evolution and the second subsystem is an ancilla qubit $\mathcal{H}_{\text{anc}}=\mathbb{C}^{2}$. 
We will identify one qubit out of that register as the control qubit, which is measured. 
The remaining $d-1$ qubits refer to the rest of the system. 
For simplicity, let us start with one system qubit ($d=1$) and one ancilla qubit so that $\mathcal{H}_{\text{2-qub.}} = \mathbb{C}^2\otimes\mathbb{C}^2$. For this minimal scenario, we recall an \textit{weak measurement} as follows \cite{koh2022experimental}:
\begin{definition}[Two-qubit null-type weak measurement]
Let $\eta\in (0,1]$. With $g(\eta):=\arcsin(\sqrt{\eta})$, the unitary operator on the ancilla qubit is defined by
\begin{equation*}
R_Y(\eta) = 
\begin{pmatrix}
  \cos g(\eta) & -\sin g(\eta) \\
  \sin g(\eta) & \cos g(\eta) \\
\end{pmatrix}.
\end{equation*}
Thus, the two-qubit controlled $R_Y$ rotation reads
\begin{align}
CR_Y(\eta) &= e^{-ig(\eta)\frac{(\mathds{1}-\sigma_z)}{2}\otimes\sigma_y} = 
\begin{pmatrix}
  \mathds{1} & 0 \\
  0 & R_Y(\eta) \\
\end{pmatrix}
\\
&\equiv 
\begin{pmatrix}
    1 & 0 & 0 & 0 \\
    0 & 1 & 0 & 0 \\
    0 & 0 & \cos g(\eta) & -\sin g(\eta) \\
    0 & 0 & \sin g(\eta) & \cos g(\eta) \\
\end{pmatrix}.
\end{align}
Applying projection operators 
\begin{equation}\label{eq:p_and_q}
P:= \begin{pmatrix}
        0 & 0 \\
        0 & 1 \\
    \end{pmatrix}, \qquad
    Q:= \begin{pmatrix}
        1 & 0 \\
        0 & 0 \\
    \end{pmatrix}
\end{equation}
to the ancilla qubit yields a weak measurement. The Kraus operators of this measurement are given in (\ref{eq:Kraus_operators}).
\end{definition}
Note that $\eta \mapsto g(\eta)$ parametrizes the angle of the $R_Y$ rotation, i.e. $g(\eta) \in (0,\frac{\pi}{2}]$.  Thus, measuring the ancilla qubit \textit{strongly} corresponds to measuring the control qubit \textit{weakly.} The setup is visualized in Fig. \ref{fig:quantum_circuit_first_return}.
\begin{figure*}
    \centering
    \includegraphics[width=\textwidth]{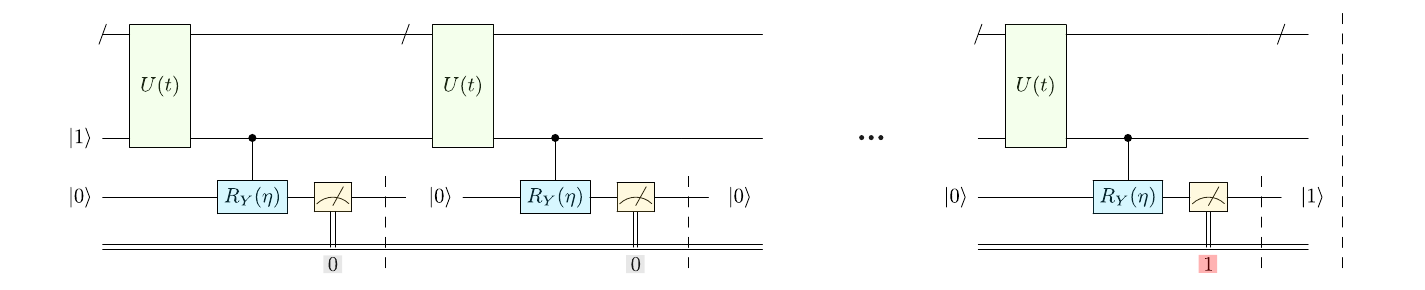}
    \caption{\justifying Quantum circuit of the first return event via weak measurements, with a renormalized rotation angle $\eta\in (0,1]$. The bit string of this measurement is $00..01$, where the $1$ occurs at position $n$ so that the circuit depth is $n\leq N$.}
    \label{fig:quantum_circuit_first_return}
\end{figure*}
W.l.o.g., we initialize the ancilla qubit in state $\ket{0}$ and the system qubit in state $\ket{1}$. One iteration in our weakly monitored quantum walk consists of three steps. At first, the unitary operator $U(t)$ is applied to the register of $d$ qubits. Secondly, the controlled $R_Y$ gate is coupled to the control qubit which is to be measured. The control qubit is part of the quantum system under investigation, where the ancilla qubit represents an \textit{environment,} through which the weak measurement is performed. Finally, the ancilla qubit is measured. If the measurement is successful (i.e. the system qubit was activated), the protocol is terminated. If the measurement was not successful, the ancilla qubit is reset to state $\ket{0}$, and the protocol continues in step one.

Hence, the measurement-induced quantum trajectory \cite{skinner2019measurement,gebhart2020topological,snizhko2021weak} on the ancilla qubit is a \textit{flip} from $\ket{0}$ to $\ket{1}$, with conditional resets performed along each unsuccessful measurement. Equivalently, this \textit{transition event} on the ancilla space $\mathcal{H}_{\text{anc}}$ corresponds to a \textit{return event} on a quantum subspace $V\subset\mathcal{H}_{\text{sys}}$. To illustrate this, recall the two-qubit CNOT gate acting on the (product) states $\ket{1}\ket{0}$ and $\ket{0}\ket{0}$:  
\begin{align*}
    \text{CNOT} (\ket{1}\ket{0}) &= \ket{1}\ket{1} \\
    \text{CNOT} (\ket{0}\ket{0}) &= \ket{0}\ket{0}. 
\end{align*}
If the ancilla qubit is in state $\ket{1}$, the particle has been detected (because the control qubit is also in $\ket{1}$). If the ancilla qubit is in state $\ket{0}$, the particle has not been detected (control qubit is in $\ket{0}$). In our implementation we use the parameterized, controlled $R_Y$ rotation $CR_Y(\eta)$. Depending on $\eta\in (0,1]$ this gate characterizes the entanglement between the control and the ancilla qubit. For $\eta=1$, it follows the same logic as the CNOT gate for the states considered above, i.e. 
\begin{align*}
     CR_Y(\eta=1) (\ket{1}\ket{0}) &= \ket{1}\ket{1} \\
    CR_Y(\eta=1)(\ket{0}\ket{0}) &= \ket{0}\ket{0}. 
\end{align*}
For general $\eta\in(0,1]$, the resulting states evolve as \cite{koh2022experimental}
\begin{align*}
&\left( \alpha \ket{0} + \beta \ket{1} \right) \ket{0} \rightarrow \\
&\left( \alpha \ket{0} + \beta \sqrt{1 - \eta} \ket{1} \right) \ket{0} + \beta \sqrt{\eta} \ket{1} \ket{1},
\end{align*}
with details given in Prop. \ref{prop:state_after_measurement}. Since the ancilla state remains always in $\ket{0}$ after every unsuccessful measurement, there is no additional mathematical constraint for the reset operation. 
By measuring the ancilla qubit, one gains the desired information about the control qubit, depending on the coupling strength $\eta$. $\eta=0$ indicates a full separation between the two qubits, i.e. no information at all (we henceforth exclude $\eta=0$ from the protocol). $\eta=1$ means that the two qubits are maximally entangled. Thus, the $\eta=1$ scenario is equivalent to taking \textit{strong} measurements, as developed in \cite{grunbaum2013recurrence, bourgain2014quantum}. In the subsequent mathematical formulation, we omit the ancilla qubit (and hence omit doubling the dimension of the Hilbert space), but work with linear shearings of the projection operators on the non-dilated Hilbert space $\mathcal{H}_{\text{sys}}$.

These shearings are seen as follows. By defining 
\begin{equation}\label{eq:peta_and_qeta}
P_{\eta}:= K^-(\eta), \qquad
    Q_{\eta}:=K^+(\eta),
\end{equation}
and writing $P_{\eta}=\sqrt{\eta}P$, $Q_{\eta} = Q+\sqrt{1-\eta}P$, with $P,Q,$ as in (\ref{eq:p_and_q}),
the Kraus operators for phase damping (\ref{eq:Kraus_operators}) are the non-dilated linear transformation of the projection operators:
   \begin{equation}\label{eq:shear_projectors}
    \begin{pmatrix}
        P_{\eta} \\
        Q_{\eta}
    \end{pmatrix}
     = 
     \begin{pmatrix}
        \sqrt{\eta} & 0 \\
        \sqrt{1-\eta} & 1
    \end{pmatrix}
     \begin{pmatrix}
        P \\
        Q
    \end{pmatrix}.
    \end{equation}
Note that such a one to one mapping between operator-valued and projection-valued measurements does not exist for other quantum channels \cite{nielsen2010quantum}.
Due to the non-projection-valued measurement, it is emphasized that $P_{\eta} + Q_{\eta} \neq \mathds{1}$, but rather (see Fig. \ref{fig:one_over_eta_half_circ_survival_ops} (b))
\begin{equation}\label{eq:pseudo_orthogonal}
P_{\eta}+Q_{\eta}=\mathds{1}+f(\eta)P,
\end{equation}
with $f(\eta)$ as in (\ref{eq:f_eta}).
This two-qubit measurement can be embedded into any register of $d$ qubits as long as one control qubit is marked.
\subsection{Weak First Return Event}
\subsubsection{Weak First Return Amplitude Operator}\label{subsubseq:weak_first_return_amplitude_operator}
We analyze the first hitting time depending on $\eta\in(0,1]$
by applying the spectral description of quantum recurrence \cite{grunbaum2013recurrence,bourgain2014quantum, grunbaum2018generalization, grunbaum2020quantum} to weak measurements. 
It is emphasized that $P_{\eta}$ and $Q_{\eta}$ characterize a binary measurement, as the small Proposition \ref{prop:yes_no_condition} shows. The logical negation of this result indicates that whenever there is a $\eta$-dependent, non-zero overlap with $S_V$, we have detected the particle and the protocol is terminated. Otherwise, the measurement process is continued.   

By applying the modified measurement operators from (\ref{eq:peta_and_qeta}) to the \textit{survival} and \textit{termination} operator, we modify the first return amplitude operator from (\ref{eq:first_return_amplitude_op}) as follows:
\begin{definition}
    Let $\ket{\psi}\in S_V$. We define the \textit{survival operator} as $\Tilde{U}_{\eta}(t):=Q_{\eta}U_t \equiv \left(Q + \sqrt{1-\eta} P\right)U_t$ and the \textit{termination operator} as $P_{\eta}U_t$. This allows us to define the \textit{weak first return amplitude operator} as
    \begin{equation}\label{eq:weak_first_return_amplitude_op}
        \Hat{a}^{\eta}_{n}(t):= P_{\eta} U(t) \Tilde{U}_{\eta}(t)^{n-1} P
    \end{equation}
\end{definition}
Consequently, the \textit{first detection probability} (\ref{return_probability}) is modified to a \textit{$\eta-$dependent first detection probability} 
\begin{equation}\label{eq:eta_return_probability}
    p_n(\eta,t):= \norm{\Hat{a}_n^{\eta}(t) \ket{\psi}}^2, 
\end{equation}
with a \textit{total detection probability}
\begin{equation}\label{eq:total_return_eta}
R_N(\eta, t) = \sum_{n=1}^{N} p_n(\eta,t),    
\end{equation}
and with $\Hat{a}_n^{\eta}(t)$ as in (\ref{eq:weak_first_return_amplitude_op}). Note that the projector $P$ on the right side of Eq. (\ref{eq:weak_first_return_amplitude_op}) establishes that $\ket{\psi}\in \text{Im} \,P.$ Do not confuse it with a weak measurement in the zero'th step. 
This total detection probability satisfies
\begin{align}\label{eq:return_prob_eta}
\begin{split}
    R_N(\eta,t) &= \sum_{n=1}^{N} p_n(\eta, t) = \sum_{n=1}^{N}  \norm{\Hat{a}^{\eta}_n(t)\ket{\psi}}^2 \\
    &= 1 - \norm{\Tilde{U}_{\eta}(t)^N\ket{\psi}}^2,
\end{split}
\end{align}
from which one concludes that 
\begin{equation}\label{eq:total_detection_prob}
    \lim_{N\to\infty}R_N(\eta, t) = 1.
\end{equation}
Details are shown in Corollary \ref{cor:total_probability_weak} and further details can be found in Appendix
\ref{subseq:appendix_weak_measurements}.
\subsubsection{First Hitting Time after $N$ Weak Measurements}
Similarly to the above weak first return amplitude operator (\ref{eq:weak_first_return_amplitude_op}) and return probability (\ref{eq:eta_return_probability}), we have for the \textit{weak first hitting time} after measuring maximally $N$ times:
\begin{equation}\label{eq:weak_first_return_time_N}
     \tau_N(\eta,t):=  \frac{\sum_{n=1}^N n \, p_n(\eta, t)}{R_N(\eta, t)}
\end{equation}
By employing the telescoping structure of the sum, we can compute a simpler expression of (\ref{eq:weak_first_return_time_N}) similar to \cite{bourgain2014quantum}, which reads
\begin{equation}\label{eq:weak_first_return_time_N_telescoping}
        \tau_N(\eta,t)=\frac{\sum_{k=0}^{N-1}\norm{\Tilde{U}_{\eta}(t)^k\ket{\psi}}^2 - N \norm{\Tilde{U}_{\eta}(t)^N\ket{\psi}}^2}{1-\norm{\Tilde{U}_{\eta}(t)^N\ket{\psi}}^2}.
\end{equation}
Details are presented in Lemma \ref{lem:return_prob_and_time} and complemented by Corollary \ref{cor:return_prob_estimates} and \ref{cor:return_time_estimates}.
Intuitively, these statements might be obvious. The weaker we measure, the smaller is the return probability and consequently the longer it will take to detect the quantum walker.
\subsubsection{First Hitting Time after Infinitely Many Measurements}
In the limit of infinitely many measurements, Eq. (\ref{eq:total_detection_prob}) yields $R_{\infty}(\eta,t)=1.$
Then, (\ref{eq:weak_first_return_time_N}) becomes
\begin{equation}\label{eq:weak_first_return_time_infty}
    \tau_{\infty}(\eta, t):= \sum_{n=1}^{\infty} n \, p_n(\eta,t) \overset{\text{def.}}{=} \sum_{n=1}^{\infty} n \, \norm{\Hat{a}_n^{\eta}(t)\ket{\psi}}^2.
\end{equation}
For this limit, the Laplace transform of the first return operator provides a useful tool \cite{grunbaum2013recurrence, bourgain2014quantum, grunbaum2018generalization, yin2019large}.
Due to the fact that $\norm{\hat{a}_n^{\eta}(t)} \leq 1 \,\forall \,\eta \in [0,1]$, one defines the (absolutely convergent) generating function
\begin{equation}\label{eq:tau_weak}
    \hat{a}(\eta, z,t):=\sum_{k=1}^{\infty} \hat{a}_k^{\eta}(t)\,z^k, \qquad z \in \mathbb{D}:=\{z\in\mathbb{C}: |z| < 1\},
\end{equation}
By using the Green's function of the survival operator,
\begin{equation}\label{eq:u_tilde_eta}
    \Tilde{U}_{\eta}(z,t):=\left(\mathds{1} - z \Tilde{U}_{\eta}(t)\right)^{-1},
\end{equation}
the generating function in (\ref{eq:tau_weak}) becomes
\begin{equation}\label{eq:generating_function}
    \hat{a}(\eta, z,t) = z P_{\eta} U(t) \Tilde{U}_{\eta}(z,t) P
\end{equation}
Thus, the complex analytic methods from \cite{grunbaum2013recurrence,bourgain2014quantum} are directly applicable. In analogy to \cite{bourgain2014quantum}, we define $z:=e^{i\theta}$ to rewrite (\ref{eq:weak_first_return_time_infty})
as the $\eta-$dependent integral, which is $\frac{-1}{2\pi}$ times a \textit{one-parameter family of Aharonov-Anandan phases} \cite{aharonov1987phase, wu1996berry, bourgain2014quantum, bengs2023aharonov}, i.e.
\begin{equation}\label{eq:aharonov_anandan_eta}
    \tau_{\infty}(\eta, t)\equiv \frac{1}{2\pi i} \int_{0}^{2\pi} d \theta \left\langle\psi(\eta, \theta, t),\frac{\partial}{\partial\theta}\psi(\eta, \theta, t)\right\rangle,
\end{equation}
with $\ket{\psi(\eta, \theta,t)}:= \hat{a}(\eta, e^{i \theta},t)\ket{\psi}$ and $\theta \in [0, 2\pi]$.
The detailed derivation of (\ref{eq:aharonov_anandan_eta}) is presented in Appendix \ref{appendix:details_integral}.
Note, that we only consider a finite-dimensional Hilbert space. The spectral measure is the point spectrum of the survival operator \cite{bourgain2014quantum, grunbaum2018generalization}. An example of (\ref{eq:aharonov_anandan_eta}) is presented in Fig. \ref{fig:2_vertex_graph} (d).

Eq. (\ref{eq:aharonov_anandan_eta}) is equivalent to computing \cite{grunbaum2020quantum} 
\begin{equation}\label{eq:vectorized_z_1_limit}
    \tau_{\infty}(\eta, t) = 1 + \lim_{x\to 1} \frac{\partial}{\partial x} \Tr\left( vec\,\Hat{a}(\eta, x, t)\,vec\, \rho \right)
\end{equation}
with $\rho:=\ketbra{\psi}{\psi}$ and $vec\,\Hat{a}(\eta, x, t)$ being the vectorized generating function of (\ref{eq:generating_function}), i.e.
\begin{equation}\label{eq:vectorized_geneating_function}
\begin{split}
    vec\,\Hat{a}(\eta, x, t) = &(P_{\eta}\otimes P_{\eta})(U_t\otimes \overline{U_t})\\
    &\left(\mathds{1}\otimes\mathds{1}-x(Q_{\eta}\otimes Q_{\eta})(U_t\otimes \overline{U_t})\right)(P\otimes P).
\end{split}
\end{equation}
The integral representation in (\ref{eq:aharonov_anandan_eta}) is formulated in the quantum system's Hilbert space, whereas (\ref{eq:vectorized_z_1_limit}) requires squaring that space. 

\section{\label{section:Indirect_First_Hitting_Time}Weakly Monitored First Hitting Time}
\subsection{Analytical Results}
Equations (\ref{eq:weak_first_return_time_N}) and (\ref{eq:aharonov_anandan_eta}) allow an analytical discussion for both the \textit{Zeno limit} and the minimal, two-level system. 
\subsubsection{Zeno Dynamics}
If the rate, at which a quantum walker is measured, is too high, there is no unitary evolution, but the walker is \textit{frozen} at its initial state. This freezing phenomenon is known as \textit{quantum Zeno effect} \cite{itano1990quantum}. Mathematically, this limit is obtained by sending $t$ to $0$, so the unitary evolution time becomes infinitely small. 
If $\ket{\psi}$ is in initial state, i.e. $P\ket{\psi} = \ket{\psi}$, then by Eq. (\ref{eq:shear_projectors}), $Q_{\eta}\ket{\psi} = \sqrt{1-\eta} \ket{\psi}$, so the first detection probability (\ref{eq:eta_return_probability}) reads 
\begin{equation}\label{eq:zeno_limit}
     p_n(\eta,0) = \eta (1-\eta)^{n-1}.  
\end{equation}
For $\eta=1$, $p_1(1,0)=1$ and  $p_n(1, 0) = 0$, for all $n> 1$.
Hence, the usual (strongly monitored) Zeno limit for the detection probability coincides with the weak protocol, i.e. when the dynamics is frozen, the detection probability in step one is 100 percent. 
For $\eta\in(0,1)$ and $n > 1$, $0\leq p_n(\eta,0) <1$, enabling a non-zero probability to detect the particle in the second, third, etc. measurement, \textit{even though} the system was frozen at its initial conditions. Since any $\eta<1$ causes a loss of information, one can interpret such a delocalization as a lack of accuracy during the measurement process. 

For the Zeno limit, we can compute the \textit{strong} first hitting time $\tau_N(1,0)$, as well as the \textit{weak} first hitting time $\tau_N(\eta,0)$ analytically. Naturally, since the quantum walker does not move at all, $\tau_N(1,0) = 1$, i.e. only one measurement is required to find the quantum walker. Due to loss of information at each measurement step in the weak protocol, this number decays reciprocal with $\eta$ as
\begin{equation}\label{eq:one_over_eta_law}
    \tau_N(\eta,0)=\frac{1}{\eta} - \frac{N (1-\eta)^N}{1-(1-\eta)^N}.
\end{equation}
Details are presented in Lemma \ref{lem:one_over_eta}. Coupling the quantum system to an environment, i.e. the ancilla qubit, corresponds to loosing information, or equivalently, missing some first return events on average. Hence, 
\begin{equation*}
\tau_{\infty}(\eta,0) = \frac{1}{\eta}>1, \qquad 0<\eta<1.
\end{equation*}
For infinitely many measurements, the projection-valued limit ($\eta=1$) is clearly recovered, i.e. $\tau_{\infty}(1,0)=1$.
\begin{figure}[ht]
    \centering
    \includegraphics[width=\linewidth]{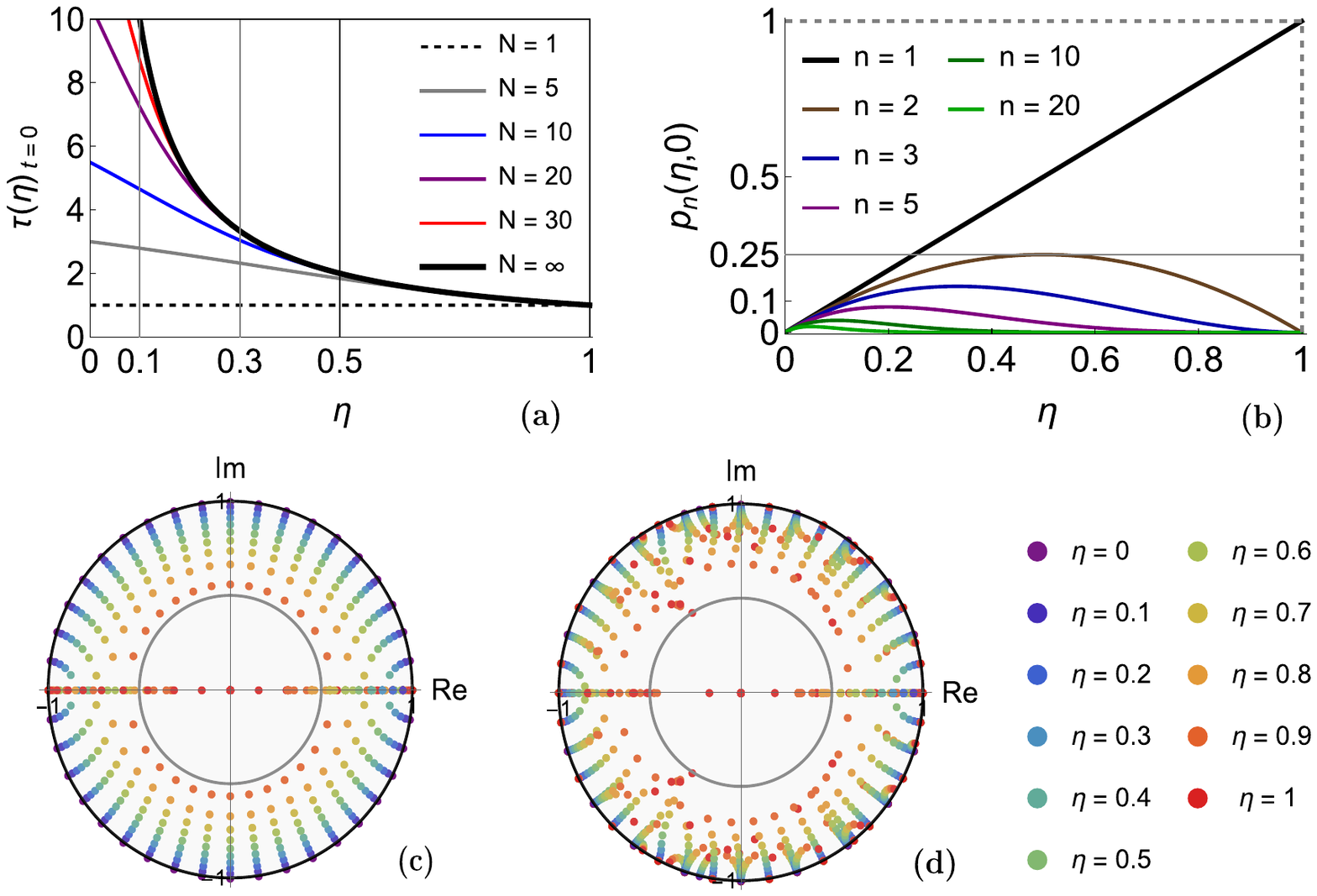}
    \caption{\justifying (a) $\frac{1}{\eta}$ decay of the mean recurrence time for $t=0$: The function $\tau_N(\eta, t=0)$ from (\ref{eq:one_over_eta_law}) is plotted against $\eta\in (0,1]$ for several values of $N$. One can see that the convergence towards infinitely many measurements happens very fast. Even for a small number of measurements, one gets the dominating $\frac{1}{\eta}$ behavior. For $N=\infty$, the $N-$ dependent corrections vanish, so that the black curve is exactly $\tau_{\infty}(\eta,t=0) = \frac{1}{\eta}$. (b) $p_n(\eta,0)$ from Eq. (\ref{eq:zeno_limit}) plotted against $\eta$ for several values of $n$. (c) $\eta$-dependent spectrum of the survival operator $\Tilde{U}_{\eta}(t)$ for the two-level system described in Sec. \ref{subseq:minimal_example}. The spectrum is shown for 40 equidistant time steps between $0$ and $2\pi$, i.e. $\Delta t = \frac{\pi}{20}.$ As $\eta$ goes to zero, the eigenvalues move towards the edge of the unit disk. At $\eta=0$, the values lie exactly on the edge, which means that no measurement intercepts at all (i.e. $\tau_N(\eta = 0,t) = \infty$ for any $N$). The inner circle has radius $\frac{1}{2}$, which corresponds to the dimensional separation of the subspace measurement. (d) $\eta$-dependent spectrum of the survival operator $\Tilde{U}_{\eta}(t)$ for the benzene ring. Due to larger system size and more support points, the spectral regularity is more disturbed compared to (c).}  \label{fig:one_over_eta_half_circ_survival_ops}
\end{figure}
\subsubsection{Two Level System}\label{subseq:minimal_example}
\begin{figure}[t]
    \centering
    \includegraphics[width=\linewidth]{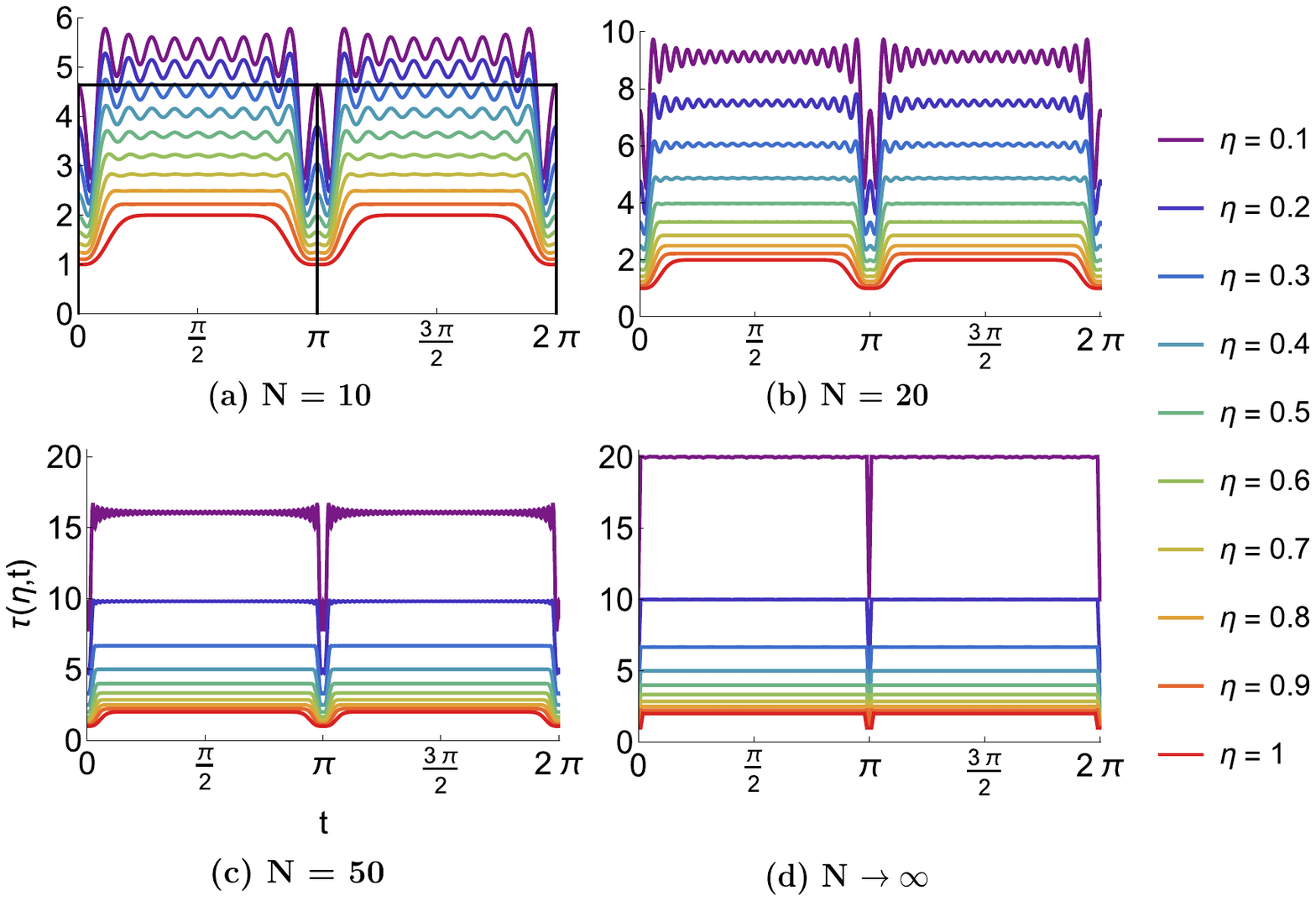}
    \caption{\justifying Two-level system: expected first return time $\tau_N(\eta,t)$ of a MCTQW on a 2-vertex graph with one edge with (a) $N=10$, (b) $N=20$, (c) $N=50$, and (d) $N=\infty$ measurements, plotted against $t\in[0,2\pi]$ for several values of $\eta$. This minimal graph is illustrated in Figure \ref{fig:ibm_q_plot}. 
    We start the protocol in vertex 0 and compute the expected number of measurements, which it takes for the particle to return for the first time. 
    $\eta=1$ corresponds to a projection-valued measurement with $\tau(1,t)=2$, except for $t=\pi$. Hence, apart from $t=\pi$, the average return time is expected to be between $2$ and $20$ (compare with (\ref{eq:two_level_sys_symbolic})). One obtains Gibbs oscillations of the return time, which increase in the amplitude as $\eta$ goes to $0$. In (a): the vertical black lines indicate the positions of the peaks that arise at multiple values of $\pi$. The horizontal black line indicates $\tau_{10}(\eta=0.1, t=0)$ in the Zeno limit ($U_0\equiv\mathds{1}$). For $\eta=1$, clearly $\tau_{10}(1,0) = 1.$ For other values of $\eta$, $\tau_N(\eta, 0)$ increases according to (\ref{eq:one_over_eta_law}).
    We see oscillations for a small number of interspersed measurements ($N\simeq 10^1$), which decay for larger $N$. This so-called Gibbs phenomenon originates from slow convergence of the respective Fourier coefficients.}
    \label{fig:2_vertex_graph}
\end{figure}
\begin{figure}[t]
    \centering
    \includegraphics[width=\linewidth]{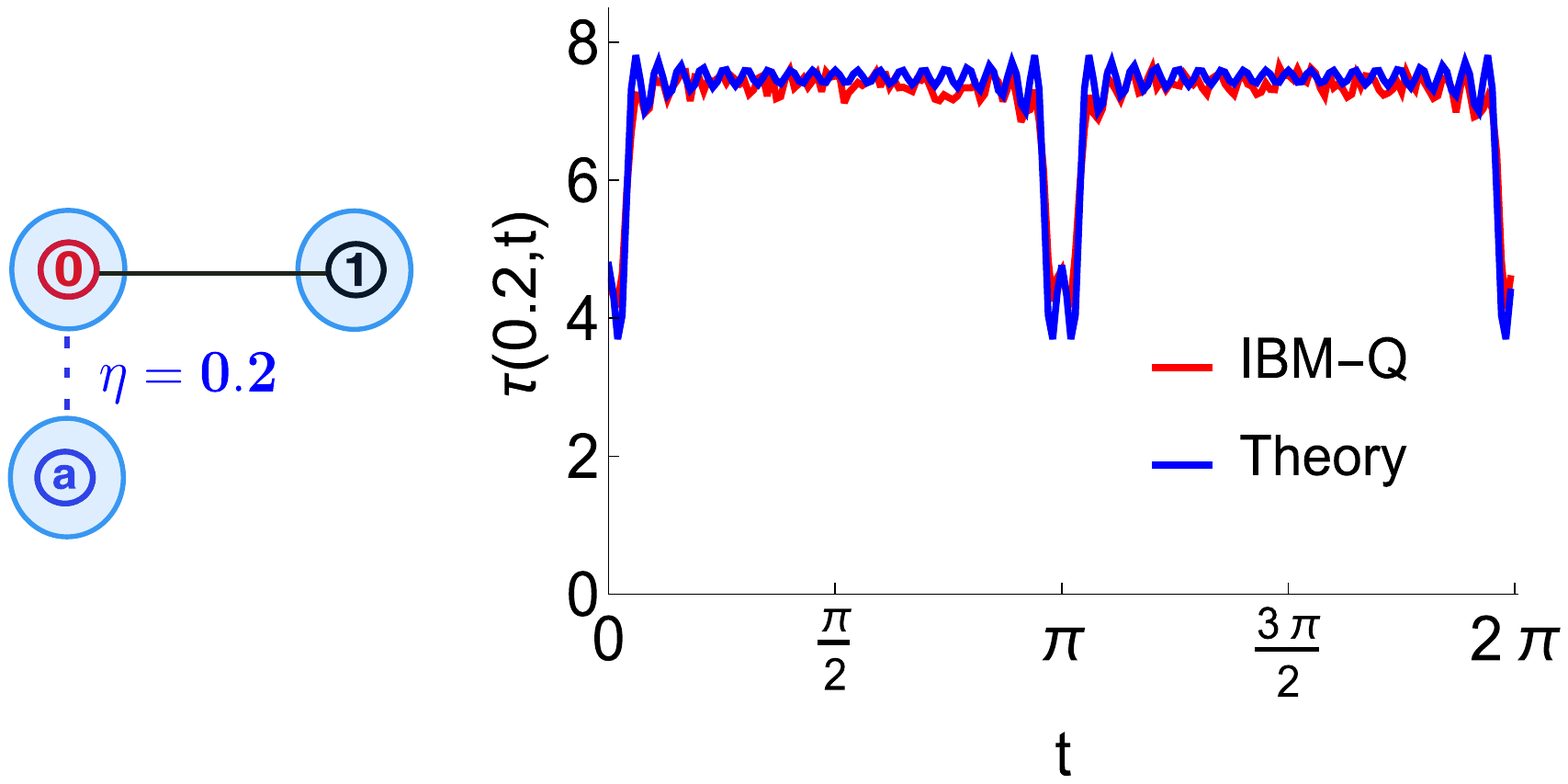}
    \caption{\justifying The expected first return time $\tau_{20}(\eta=0.2, t)$ through $N=20$ weak measurements with coupling strength $\eta=0.2$  for the minimal graph example, computed on IBM-Q Sherbrooke. Blue curve: theoretical prediction for $\tau_{20}\left(\eta=0.2, t\right)$, with $\ket{\psi}=\ket{10}$ via Eq. (\ref{eq:weak_first_return_time_N_telescoping}). Red curve: results of IBM-Q. The quantum circuit is displayed in Figure \ref{fig:quantum_circuit_first_return}. The deviation between theory and experiment is in agreement with the error rates provided by IBM. 
    }
    \label{fig:ibm_q_plot}
\end{figure}
The minimal example consists of a graph with two vertices and one edge connecting them. Without loss of generality, we measure the vertex labeled with 0, as displayed in Fig. \ref{fig:ibm_q_plot}. The Hamilton operator is given by $H = -\sigma_x \equiv - \begin{pmatrix}
    0 & 1 \\
    1 & 0 \\
\end{pmatrix}$.
The two non-degenerate eigenvalues are $E_{\pm} = \pm 1$, which determine the unitary evolution to be $U(t) = e^{-itH} = e^{-it}P_+ + e^{it} P_-$, where $P_{\pm}$ denote the projectors onto the respective eigenspaces. 
%
%
%
For a two-level system, we compute $\tau_{\infty}(t, \eta)$ by Eqs. (\ref{eq:u_tilde_eta}, \ref{eq:generating_function}, \ref{eq:vectorized_z_1_limit}). We use the vectorization approach of \cite{grunbaum2018generalization} and use $\rho:=\ketbra{\psi}{\psi},$ 
\begin{align*}
    U_t &= \begin{pmatrix}
        \cos{t} & i\sin{t} \\
        i\sin{t} &\cos{t}
    \end{pmatrix},
    \qquad
    \ket{\psi} = \begin{pmatrix}
        1 \\
        0
    \end{pmatrix},
    \\
    P_{\eta} &= \begin{pmatrix}
        \sqrt{\eta} & 0 \\
        0 & 0
    \end{pmatrix}, 
    \qquad
    Q_{\eta} = \begin{pmatrix}
        \sqrt{1-\eta} & 0 \\
        0 & 1
    \end{pmatrix}, 
\end{align*}
and
\begin{align*}
    \Tilde{U}_{\eta}(t) &= \begin{pmatrix}
        \sqrt{1-\eta}\cos t & i\sqrt{1-\eta}\sin t \\
        i \sin t &  \cos t 
    \end{pmatrix}
\end{align*}
to evaluate $g(t,x,\eta):=  \Tr\left( vec\,\Hat{a}(t,x,\eta)\,vec\, \rho \right).$ Then, 
\begin{equation}\label{eq:two_level_sys_symbolic}
\begin{split}
    \tau_{\infty}(t, \eta) &= 1 + \lim_{x\to1} \frac{\partial}{\partial x} g(t,x,\eta) \\
    &= \frac{\tau_{\infty}(t)}{\eta} = \begin{cases}
        \frac{1}{\eta}, \quad \text{if\,} t = k\pi, \,k\in\mathbb{Z}, \\
        \frac{2}{\eta}, \quad \text{else.}
    \end{cases}
\end{split}
\end{equation}
$g(t,x,\eta)$ is written out in Appendix \ref{appendix:symbolic_expression}.

%
%
%
%
Our findings coincide with those first hitting times, that have been calculated already in \cite{friedman2017quantum}. There, it was shown that
\begin{equation*}
     \tau_{\infty}(t) = \begin{cases}
        1, \quad \text{if\,} t = k\pi, \,k\in\mathbb{Z}, \\
        2, \quad \text{else.}
    \end{cases}
\end{equation*}
Their result corresponds exactly to the projection-valued limit $\eta=1$.
Consequently, our weak protocol yields for a two level system:
\begin{equation}\label{eq:1_over_eta_projection_suggestion}
    \tau_{\infty}(t, \eta) =  \frac{\tau_{\infty}(t)}{\eta}. 
\end{equation}

More generally, consider a tight-binding Hamilton operator for rings of length $L$:
\begin{equation}\label{eq:Hamiltonian_ring_L}
        H = - \sum_{i=0}^{L-1} \left(\ketbra{i}{i+1} + \ketbra{i+1}{i} \right)
\end{equation}
For this system, the mean first detection time $\tau_{\infty}$ reads (see Eq. (57) in \cite{friedman2017quantum})
\begin{equation}\label{eq:first_return_ring_L}
    \tau_{\infty}(\eta=1) = \begin{cases}
        \frac{L+2}{2}, \quad \text{if\,} L \text{\,is even,} \\
         \frac{L+1}{2}, \quad \text{if\,} L \text{\,is odd.} 
    \end{cases}
\end{equation}
Hence, for all non-characteristic times $t\in[0,2\pi),$ we expected a mean value of two, which is confirmed by the red curves ($\eta=1$) in Fig. \ref{fig:2_vertex_graph}. At those times $t\in[0,2\pi)$, that belong to the spectral support of the system, one observes discontinuous jumps in the return time. The height of the plateaus contains both information about the size of the system and also the spectral degeneracy of the underlying graph.  
Rings of size $L$ are studied in great generality in \cite{friedman2017quantum}. The \textit{exceptional sampling times} $t_e$ are given by \cite{friedman2017quantum},
\begin{equation}\label{eq:exceptional_sampling_times_ring_length_L}
t_e = \frac{2\pi k}{\Delta E} \mod 2\pi
\end{equation}
for $k\in\mathbb{Z}$ and energy difference $\Delta E \equiv E_i-E_{j}>0.$ 
In our minimal example, $L=2$ so there is only one energy gap $\Delta E = +1- (-1) = 2$. Thus, the exceptional sampling times occur at integer multiples of $\pi$, which is seen in Fig. \ref{fig:2_vertex_graph}.

We emphasize that the weak measurement procedure does not have any influence on the exceptional sampling times or the relative shape of the peaks. It is merely an offset in the first detection time $\tau$, due to a decreasing coupling strength between the ancilla qubit and the quantum system. The special sampling times are degeneracies in the spectrum of the unitary operator \cite{liu2020quantum}.
\begin{figure}[t]
    \centering
    \includegraphics[width=\linewidth]{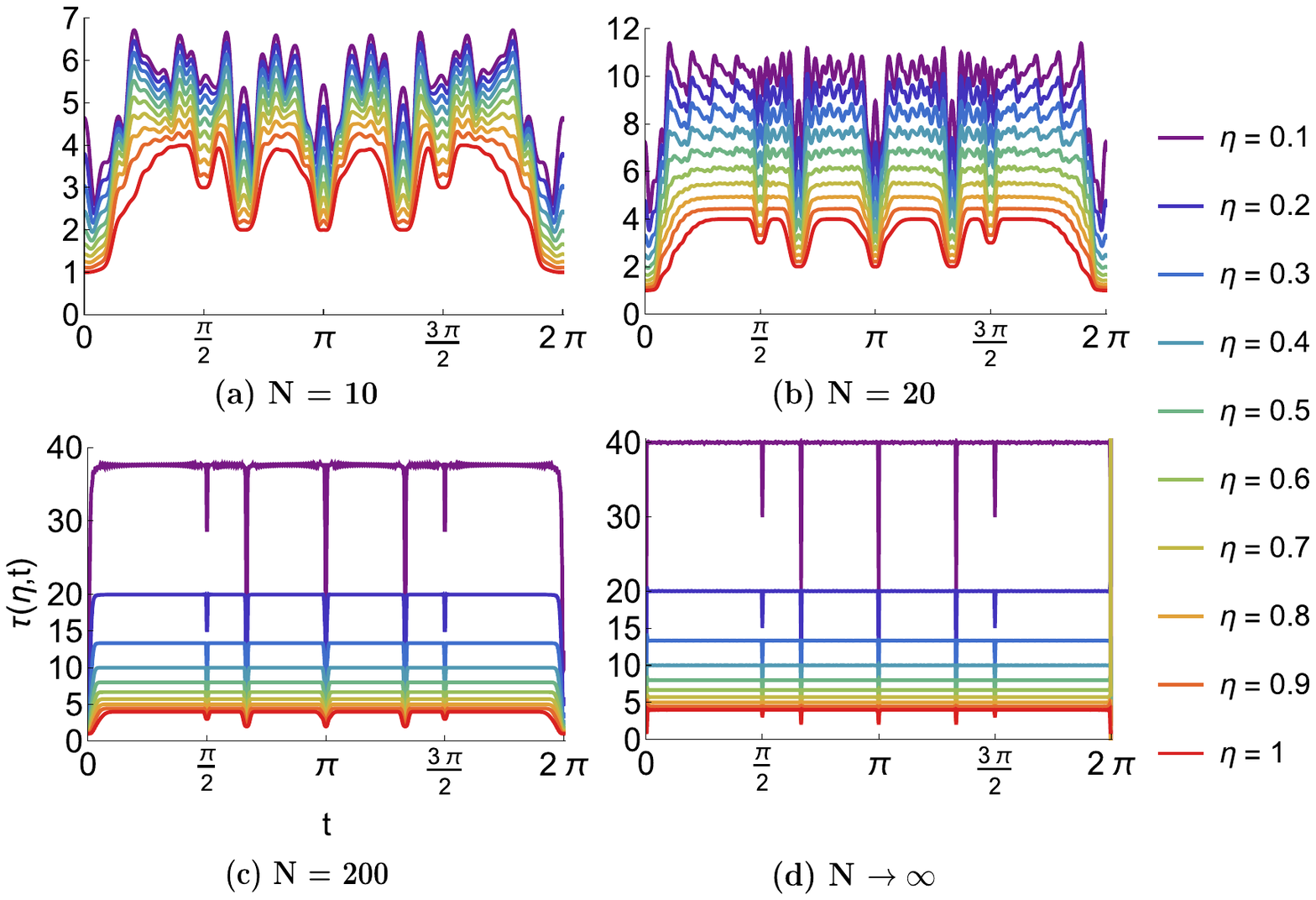}
    \caption{\justifying Benzene ring: expected first return time of a MCTQW on Benzene with (a) $N=10$, (b) $N=20$, (c) $N=200$, and (d) $N=\infty$ measurements for several values of $\eta$. 
    We start the protocol in vertex 5 and compute the expected first return time
    $\tau(\eta, t)$ against $t$ for several $\eta$. $\eta=1$ corresponds to a projection-valued measurement with $\tau(1,t)=4$, except for the exceptional points $\mathcal{T}_{\text{Benzene}}$ (as discussed in the main text). Hence, the average return time is expected to be between $4$ and $40$ (compare with (\ref{eq:one_over_eta_law})). Again, there are Gibbs oscillations of the return time, which increase in the amplitude as $\eta$ goes to $0$. For $\eta=1$, this is clearly at $\tau(1,0) = 1.$ For other values of $\eta$, that goes up according to (\ref{eq:one_over_eta_law}).
    We see oscillations for a small number of interspersed measurements ($N\simeq 10^1$), which decay for larger $N$. This (Gibbs) phenomenon is caused by a slow convergence of the underlying Fourier coefficients.}
    \label{fig:benzene_ring_graph}
\end{figure}
\subsection{Numerical and Experimental Results}
\subsubsection{Two Level System on an IBM Quantum Computer}
%
%
%
We implement the quantum circuits shown in Fig.~\ref{fig:quantum_circuit_first_return} on a superconducting quantum processor for the two-level system. The readout assignment error---which quantifies the inaccuracy in qubit state discrimination---is approximately \(5 \times 10^{-3}\). The error rates for two-qubit cross-resonance gates range between \(5 \times 10^{-3}\) and \(1 \times 10^{-2}\). To mitigate decoherence effects, we employ dynamical decoupling \cite{ezzell2023dynamical}, a technique that preserves quantum coherence by applying a sequence of control pulses to suppress environmental noise and system-bath interactions during measurement.
The experimental results obtained from IBM-Q Sherbrooke (red curve) are in excellent agreement with the theoretical predictions (blue curve) in Fig.~\ref{fig:ibm_q_plot}, showing only minor deviations and confirming the accuracy of the model. Moreover, the magnitude of these deviations is predictable by the error rates of the cross-resonance gates: the average deviation between theoretical and experimental data fluctuates around 2\%, which corresponds roughly to \(N \times10^{-3}\), where $N=20$ is the circuit depth. 
\subsubsection{Benzene Ring}\label{subseq:benzene_ring}
The first return time problem for strongly monitored quantum walks on Benzene-type rings was extensively studied in \cite{friedman2017quantum}. The complete graph $K_6$ subjected to an open quantum dynamics was studied in \cite{sinkovicz2015quantized}.
In order to make an adequate comparison to our weak measurement protocol, we choose the Benzene ring of length $L=6$ as the second example. The Hamiltonian is given in (\ref{eq:Hamiltonian_ring_L}).
The energy eigenvalues of the Hamiltonian are $\sigma(H)=\{-2,-1_2, 1_2, 2\}$, where the subscript $2$ indicates double degeneracy. 
For $L=6$, Eq. (\ref{eq:first_return_ring_L}) predicts on average $\tau_{\infty} = 4$ for nearly all values of $t$, which corresponds to the expected first return time through \textbf{strong} measurements ($\eta=1$, compare with red curves in Fig. \ref{fig:benzene_ring_graph}).
According to (\ref{eq:exceptional_sampling_times_ring_length_L}), we expect exceptional sampling times at (see Figure \ref{fig:benzene_ring_graph}) $t\in \mathcal{T}_{\text{Benzene}}:= \left\{ \frac{\pi}{2}, \frac{2\pi}{3},\pi,\frac{4\pi}{3},\frac{3\pi}{2}  \right\}$ with discontinuous return time due to fractional revival of the wave packet \cite{friedman2017quantum}.
\begin{figure}[t]
    \centering
    \includegraphics[width=\linewidth]{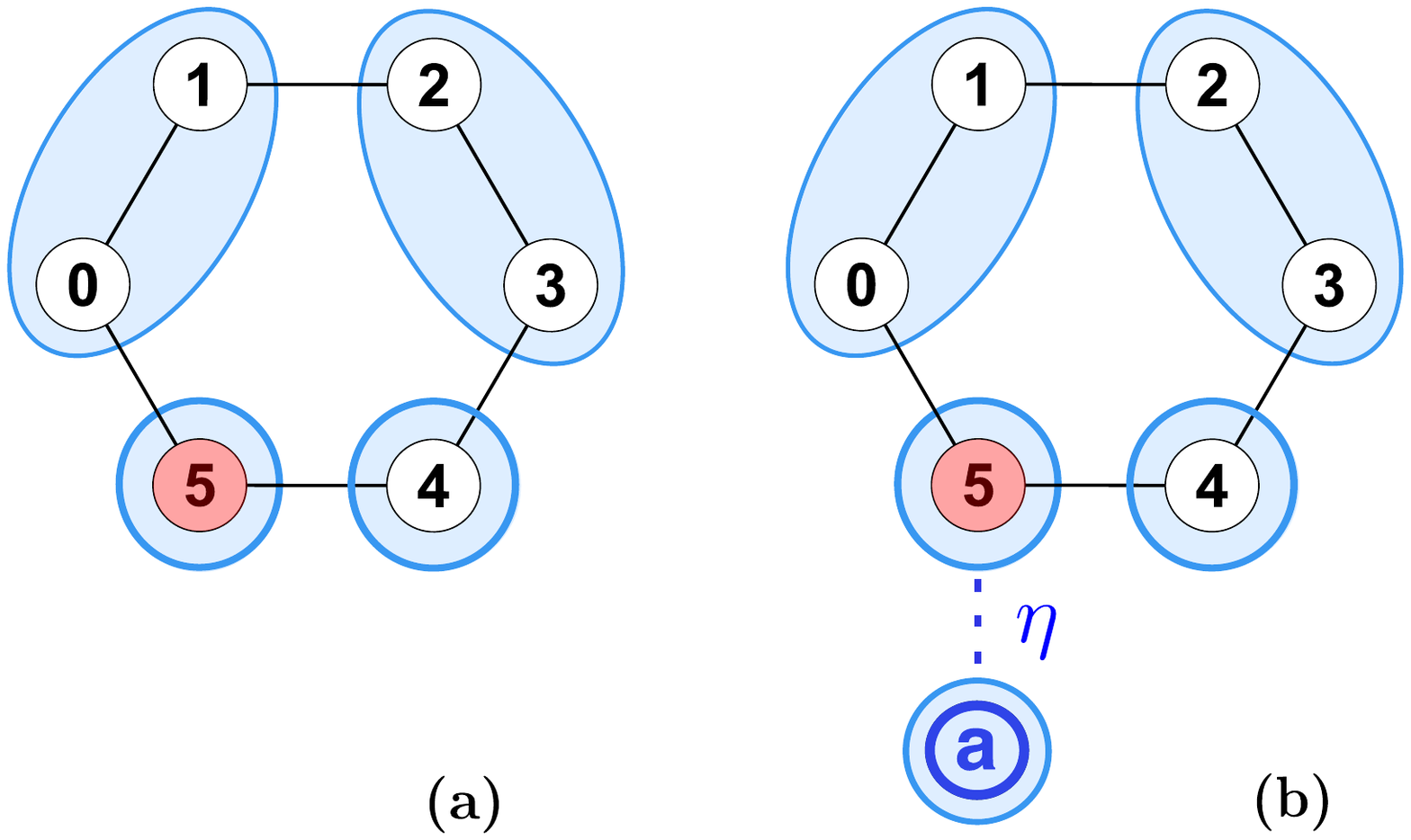}
    \caption{\justifying (a) Benzene ring encoded into a 4-qubit basis. This is the (quantum) system on which a MCTQW is performed. The blue areas indicate one qubit each. The initial vertex, which is being measured is marked in red. Since benzene is symmetric, it does not play any role, which vertex we measure. (b) The benzene ring from (a) with an additional ancilla qubit "a". The weak measurement procedure on the quantum system is implemented through the ancilla qubit with the coupling strength $\eta$. The encoding follows Proposition \ref{prop:canonical_encoding}.}
    \label{fig:benzene_ring_encoding}
\end{figure}
The resulting curves in Fig. \ref{fig:benzene_ring_graph} show qualitatively the same offset which is proporitional to $\frac{1}{\eta}$, as those in the two level system (Figure \ref{fig:2_vertex_graph}). Again, the exceptional sampling times $\mathcal{T}_{\text{Benzene}}$ are not affected by the \textit{type} of measurement protocol: whether we measure \textit{strongly} ($\eta=1$) or \textit{weakly} ($0<\eta<1$), makes no difference for the location of exceptional sampling times. This was to be expected, since the weakening of the measurement strength should not affect the spectral characteristics of the unitary operator \cite{liu2020quantum}. As $\eta$ goes to zero, the according loss of information causes an increase of the expected first return time.
Figure \ref{fig:benzene_ring_graph} (d) suggests that the first hitting time after $N\to\infty$ measurements for Benzene, $\tau^{\text{B}}_{\infty}(t, \eta)$, holds
\begin{equation*}
    \tau^{\text{B}}_{\infty}(t, \eta) =  \frac{\tau^{\text{B}}_{\infty}(t)}{\eta} = \frac{1}{\eta} \begin{cases}
        1, \quad t = 0, \\
        2, \quad t\in \left\{\frac{2\pi}{3},\pi,\frac{4\pi}{3}\right\}, \\
        3, \quad t\in \left\{\frac{\pi}{2}, \frac{3\pi}{2}\right\}, \\
        4, \quad \text{else,} \\
    \end{cases}
\end{equation*}
where every $t$ is modulo $2\pi$. 
\subsection{The $\frac{1}{\eta}$-Decay of the First Hitting Time}\label{subsec:one_over_eta_deviation}
\begin{figure}[t]
    \centering
    \includegraphics[width=\linewidth]{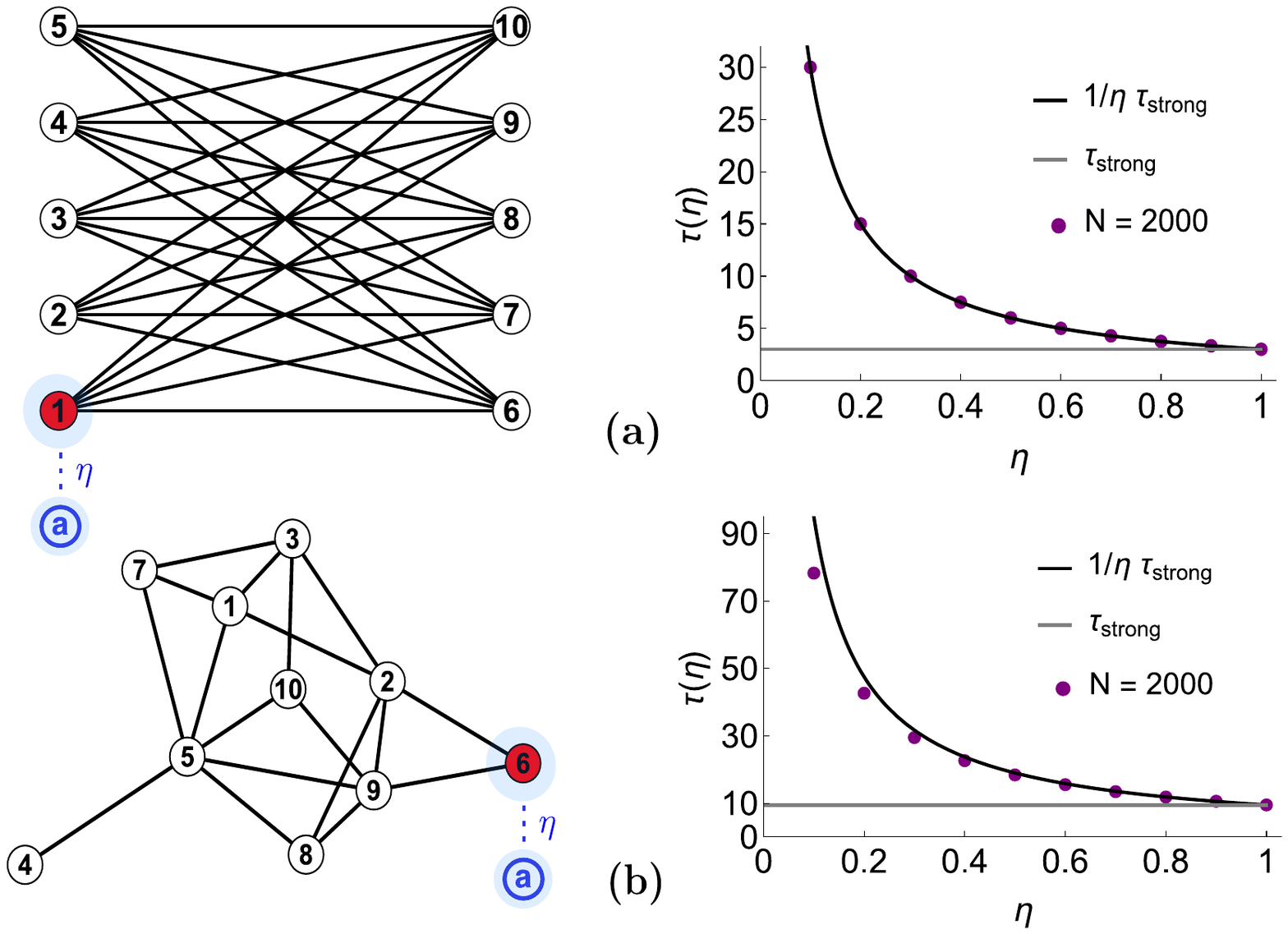}
    \caption{\justifying More examples for the $\eta^{-1}$ decay of the first hitting time with weak measurements: (a) complete bipartite graph with 10 vertices and (b) random graph with 10 vertices. The initial sites are marked in red. The blue surface indicates one qubit which is coupled to one ancilla qubit. 
    Again, the evolution time of the unitary is chosen to be fixed at $t=0.9$. We numerically compute the projection-valued first hitting time $\tau_N^{\text{strong}}$ (horizontal gray line) and then compare it with the weak first hitting time  $\tau_N$ (\ref{eq:weak_first_return_time_N_telescoping}), following the circuit in Fig. \ref{fig:quantum_circuit_first_return}, with $N=2.000$ for 10 values of $\eta$ between $\eta=0.1$ and $\eta=1.0$ (purple dots). For these systems with the respective initial conditions, the strong hitting times are: (a) $\tau_N^{\text{strong}} = 3$ and (b) $\tau_N^{\text{strong}} = 9.5$ (which goes to $\tau_{\infty}^{\text{strong}}=10$).
    The correlation in Eq. (\ref{eq:conjecture_one_over_eta}) is visible, although $\tau_N^{\text{weak}}(\eta)$ converges slower to $\frac{1}{\eta}\tau_N^{\text{strong}}$ for the random graph (b) compared to the complete bipartite graph (a).}
    \label{fig:random_graph_complete_graph}
\end{figure}
Numerical computation showed that the leading order in which the expected return time decays as $\eta$ goes from $1$ to $0$, behaves like $\frac{1}{\eta}.$ In the limit of $N\to\infty$, the $\frac{1}{\eta}$ behavior is very sharp, which indicates
$\tau_{\infty}(\eta, t) =\frac{1}{\eta} \,\tau^{}_\infty(t,1),$
where $\tau^{}_\infty(t,1)$ denotes the expected return time through strong, projection-valued measurements, i.e. $\eta=1$.  
Recall that the mean recurrence time for rank one projections on localized initial states - as it was shown by Grünbaum et al. \cite{grunbaum2013recurrence} - is always an integer or infinite. If the subspace of the measurement is higher dimensional, then the mean recurrence time is a fraction, as it was shown by Bourgain et al. \cite{bourgain2014quantum}. 
We have strong analytical, numerical and quantum-compu\-tational evidence for
\begin{equation*}\label{eq:conjecture_one_over_eta}
    \tau^{\text{weak}}_{\infty} = \frac{1}{\eta} \tau^{\text{strong}}_{\infty} \tag{$\star$}
\end{equation*}
but a rigorous proof for all unitary operators would be desirable. 
We are able to show (\ref{eq:conjecture_one_over_eta}) in the Zeno limit (Eq. (\ref{eq:one_over_eta_law}), Lemma \ref{lem:one_over_eta}), and symbolically for the two-level system (\ref{eq:two_level_sys_symbolic}). 
Several examples, which underscore Eq. (\ref{eq:conjecture_one_over_eta}) are shown in Fig. \ref{fig:random_graph_complete_graph} and Fig. \ref{fig:grid_graph_and_petersen_graph_one_over_eta}. It is mentioned that Eq. (\ref{eq:conjecture_one_over_eta}) is not restricted to fully localized initial states as it is seen in Fig. \ref{fig:grid_graph_and_petersen_graph_one_over_eta}.
\begin{figure}[t]
    \centering
    \includegraphics[width=\linewidth]{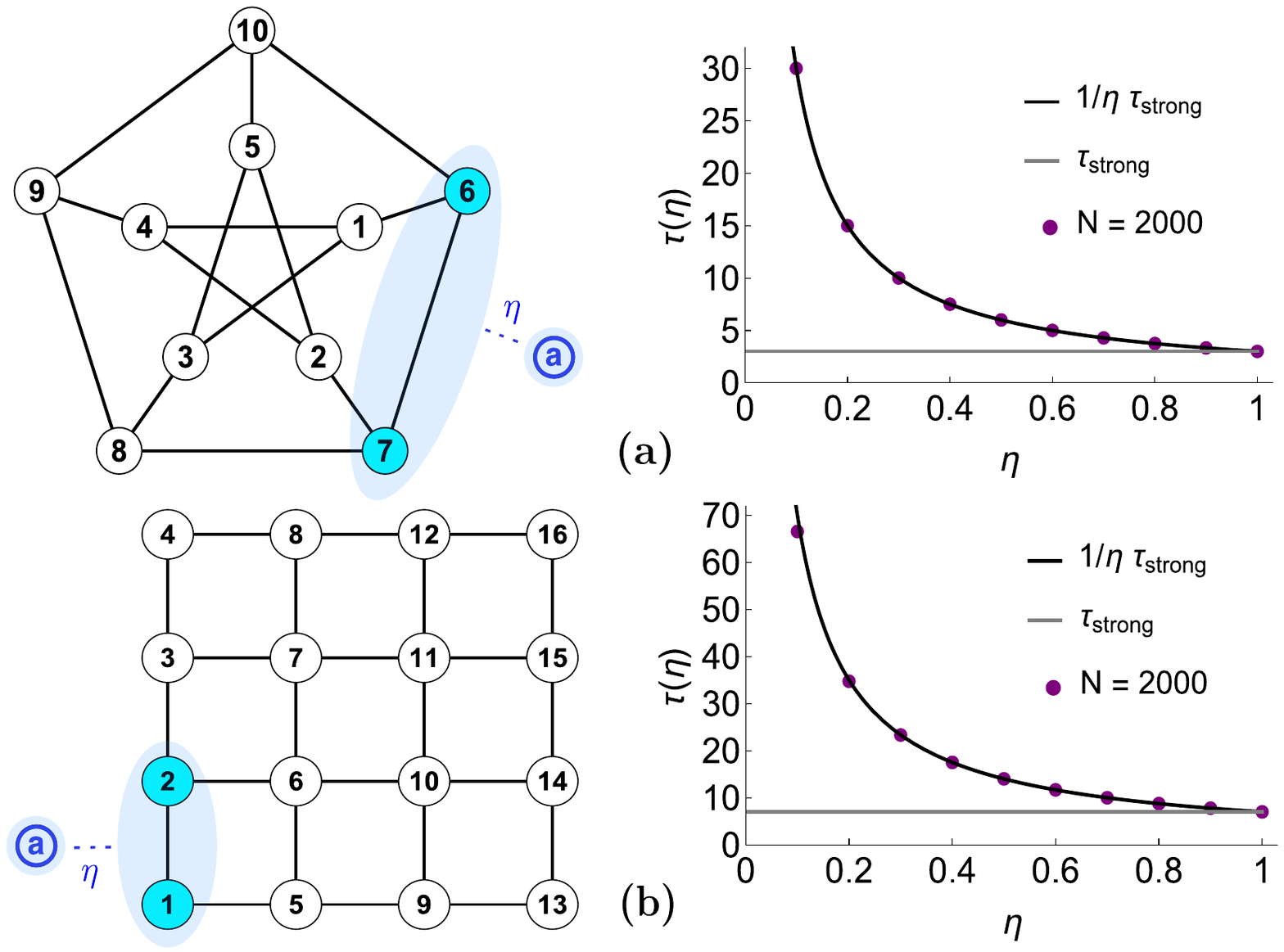}
    \caption{\justifying Examples for the $\eta^{-1}$ decay of the first hitting time with weak measurements for initial states that are in superposition of two sites: (a) Petersen graph with 10 vertices and (b) graph of a $4\times 4$ grid. The sites that span the initial state are marked in cyan. The blue surface indicates one qubit which is coupled to one ancilla qubit.  We choose a uniform superposition to be the initial state, i.e. for (a), $\ket{\psi} = \frac{1}{\sqrt{2}} \left(\ket{6} + \ket{7}\right)$ and for (b), $\ket{\psi} = \frac{1}{\sqrt{2}} \left(\ket{1} + \ket{2}\right)$.
   Again, the evolution time of the unitary is chosen to be fixed at $t=0.9$. Similar to the above examples, we numerically compute the projection-valued first hitting time $\tau_N^{\text{strong}}$ (horizontal gray line) and then compare it with the weak first hitting time $\tau_N$ (\ref{eq:weak_first_return_time_N_telescoping}), following the circuit in Fig. \ref{fig:quantum_circuit_first_return}, with $N=2.000$ for 10 values of $\eta$ between $\eta=0.1$ and $\eta=1.0$ (purple dots). For these systems with the respective initial conditions, the strong hitting times are: (a) $\tau_N^{\text{strong}} = 3$ and (b) $\tau_N^{\text{strong}} = 7.$  The expected correlation from Eq. (\ref{eq:conjecture_one_over_eta}) is visible.}
    \label{fig:grid_graph_and_petersen_graph_one_over_eta}
\end{figure}

The increase in the expected-first return time, which is proportional to  $\frac{1}{\eta}$ gives rise to the question: which $\eta'\in (0,1]$ characterizes the threshold between the quantum and the classical expected first return time? For the generic case of rings of size $L$, and for $N\to\infty$ measurements, we can give a simple answer: 
\begin{equation}\label{eq:quantum_classical_transition_rings_L}
    \eta'(L) = \frac{1}{2} + \begin{cases}
        \frac{1}{L}, \quad \text{if\,} L \text{\, is even,} \\
        \frac{1}{2L}, \quad \text{if\,} L \text{\, is odd,} \\
    \end{cases}
\end{equation}
which goes to $\frac{1}{2}$ as $L\to\infty.$
This can be seen as follows: the classical return time for those rings was shown to be $\tau^{\text{cl.}}_{\infty} = L$ \cite{friedman2017quantum}. The expected first return time through projection-valued measurements was derived in \cite{friedman2017quantum} and is given in (\ref{eq:first_return_ring_L}). Based on (\ref{eq:conjecture_one_over_eta}), we can solve $\frac{1}{\eta'} \tau^{\text{strong}}_{\infty}\overset{!}{=} \tau^{\text{cl.}}_{\infty}$ for $\eta'\in (0,1]$ and conclude (\ref{eq:quantum_classical_transition_rings_L}).
We emphasize that Eq. (\ref{eq:conjecture_one_over_eta}) remains valid, even if the time-reversal symmetry is broken through the presence of a weak, homogeneous magnetic field \cite{cresti2021convenient}. For rings, these systems are modeled by a tight-binding Hamiltonian of the form \cite{wang2024first} $H = - \sum_{j=0}^{L-1} \left( e^{i\alpha} \ketbra{j}{j+1} + e^{-i\alpha} \ketbra{j+1}{j}   \right),$ where $e^{i\alpha}$, for $\alpha\neq 1$, is a Peierls phase \cite{cresti2021convenient} which represents the magnetic flux, that causes different degeneracies in the energy levels and hence affects the mean recurrence time \cite{wang2024first}.

\section{\label{section:Conclusions}Conclusions}
In this work, we analyzed the first hitting time problem for continuous-time quantum walks on graphs through weak measurements. 
The composite quantum system under investigation was chosen to be suitable for implementation on quantum computers. The graph is encoded into a finite qubit register and an ancilla qubit is connected to \textit{one} of the system qubits. The weak measurement is implemented by a controlled rotation gate, whose reparameterized rotation angle indicates the \textit{measurement strength.} This specific operator-valued measurement corresponds to quantum evolution through a phase damping channel \cite{nielsen2010quantum}. 
The monitoring process of the quantum walk is realized by an on-site measurement protocol at a constant rate. After an initially fixed time period, the quantum system is intercepted and weakly measured. Depending on the outcome of this measurement, the protocol either \textit{continues} for no-click measurement, or is \textit{terminated}  for click measurement.
By applying the theory of quantum state recurrence \cite{grunbaum2013recurrence} and quantum subspace recurrence \cite{bourgain2014quantum}, we can compute the first hitting time depending on the rotation angle of the measurement.
Both for finitely and infinitely many measurements, we can evaluate the weakly monitored first return time by utilizing the Hilbert space formalism without dilation: a special case, which fails in general for operator-valued measurements.

Unlike in the seminal work of Grünbaum et al. \cite{grunbaum2013recurrence}, where the mean return time is strictly integer quantized, or infinite, in our setting the mean return time is not quantized. Nevertheless, a quantized structure emerges through a simple relation:
$\tau_{\text{weak}} = \frac{1}{\eta}\tau_{\text{strong}}$ 
where $\eta = \sin^2(g)$ and $g$ being the rotation angle inducing the strength of the measurement. This relation reveals a direct connection between recurrence properties under strong projective measurements and those under weak measurements through phase damping channels. While we expect such a correspondence to generally hold for recurrence observables, its validity for other classes of observables remains an open question. The relation shows, as expected, that the weaker we measure, the longer it takes to detect the recurrence. Thus, on the one hand, weak measurements interfere less with the system. On the other hand, we need more measurements to detect the system in the target state.

Consequently, several natural research questions arise. At first, can we find a similar dependence on the first hitting time by the measurement parameter if we use a different coupling, e.g. an $XX$ coupling? If yes, is it also proportional to the inverse of the parameter?
Secondly, a generalization to multi-qubit weak measurements with various rotation angles provides a promising area of research with applications in designing measurement-induced quantum algorithms. 
Thirdly, our implementation draws a bridge between recurrence and the rich literature on entanglement measures \cite{vidal2002computable, vedral1997quantifying, eisert1999comparison} and information theory of networks \cite{radicchi2020classical}. Studying the connection between the first hitting time and quantum non-locality may provide further insights into the interplay between stochastic processes, entanglement theory, information theory and foundational quantum mechanics. 
Finally, the analogous \textit{first transfer time problem} may be studied within the same scope, both for phase damping channels and more general quantum measurements.

\subsection*{\label{section:Acknowledgements}Acknowledgements}
%
TH gratefully acknowledges Alessio Belenchia, Benjamin Desef, Maxim Efremov, Jakob Murauer, Felix Rupprecht, and Sabine Wölk for valuable discussions.  
This project was made possible by the DLR Quantum Computing Initiative and the Federal Ministry for Economic Affairs and Climate Action.
%
EB acknowledges the support of Israel Science Foundation's grant 1614/21.
%
%
We acknowledge the use of IBM Quantum services for this work. The views
expressed are those of the authors and do not reflect the official policy or position of IBM or the IBM Quantum team. In this paper, we used the IBM Sherbrooke Processor, which is an IBM Quantum
Eagle Processor. 

\bibliographystyle{ieeetr}
\bibliography{main}

\onecolumngrid
\appendix
\section{Recurrence by Weak Measurements}
\subsection{Encoding of the First Return Event}\label{subseq:appendix_practical_encoding}
In order to implement monitored quantum walks on real quantum hardware, one must give a precise encoding of states into qubit-Hilbert spaces of dimension $d = 2^N$. The goal of such an embedding is to keep $N\in\mathbb{N}$ as small as possible in order to minimize computational resources. The following well-known block encoding ensures the usage of $\log(d)$ many qubits for computing first hitting times on rings. 
\begin{proposition}\label{prop:canonical_encoding}
    Consider the tight-binding Hamiltonian of the form (\ref{eq:Hamiltonian_ring_L}) for rings of length $L\in\mathbb{N}$. If $k<L$ sites are to be weakly measured, then there exists an encoding which requires $ N(L) = \lceil \log_2(L-k) \rceil + 2k $ many qubits. 
\end{proposition}
\begin{proof}
Our goal is to keep the number of qubits $N\in\mathbb{N}$ as small as possible, where $N$ is a function of the ring length $L$. 
Assuming, we measure $k<L$ qubits weakly, we need an additional register of $k$ ancilla qubits (one ancilla per site). We call the number of sites that we do not measure $n:=L-k$, thus $L=n+k.$ For the first $n$ sites we can choose the canonical bit-to-qubit mapping to occupy n modes. That is, take the binary representation of each $i\in [n]$, i.e. $0 = \ket{000..0_n 000...00}$, $1=\ket{000....1_n 000...00}$, $2 = \ket{00...10_n 000...00}$ and so on. In order to get the number of qubits needed for this, just compute $\lceil \log_2(n) \rceil = \lceil \log_2(L-k) \rceil$. 
Since there are $2k$ qubits left, the number of required qubits is $N(L) = \lceil \log_2(L-k) \rceil + 2k.$
\end{proof}
\subsubsection{State after the weak measurement}
\begin{figure}[h]
     \centering
     \begin{subfigure}[b]{0.496\textwidth}
         \begin{align*}
        \Qcircuit @C=1.5em @R=2.7em {
        \lstick{\ket{\phi}} & 
        \gate{U(t)} & \meter \cwx[1] \barrier[-0.6em]{1} & \qw \\ 
        & \cw & \cw & \cw \\
        }
        \end{align*}
        \caption{\justifying One iteration of the MCTQW with a projection-valued measurement on the system. We refer to this type of measurement as a \textbf{strong} measurement.}
        \label{fig:one_iteration_projective_ctqw}
     \end{subfigure}
     \hfill
     \begin{subfigure}[b]{0.496\textwidth}
         \begin{align*}
        \Qcircuit @C=1.1em @R=2.3em {
        \lstick{\ket{\phi}} & \gate{U(t)} & \ctrl{1} & \qw & \qw \barrier{2} & \qw &  \\
        \lstick{\ket{0}} & \qw & \gate{R_Y(\eta)}  & \meter \cwx[1] & \qw & \qw & \\ 
         & \cw & \cw & \cw \cwx[-1] & \cw & \cw \\
        }
        \end{align*}
        \caption{\justifying One iteration of the MCTQW with a weak measurement: projection-valued measurement on the ancilla system corresponds to a weak measurement on the quantum system. We refer to this type of measurement as a \textbf{weak} measurement.}
        \label{fig:one_iteration_weak_ctqw}
     \end{subfigure}
    \caption{(a) \textbf{Strong} and (b) the \textbf{weak} measurements.}
\end{figure}
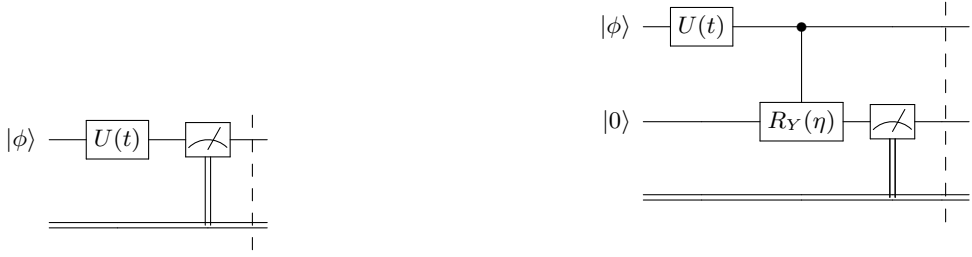
We consider a strong (Figure \ref{fig:one_iteration_projective_ctqw}) and an weak measurement (Figure \ref{fig:one_iteration_weak_ctqw}). Let us consider the post-measurement state. 
Let $\ket{\psi(t)} = U(t)\ket{\phi}_{\text{ctrl-qubit}}$ denote the $t$-dependent state after one iteration of $U(t)$ on the control qubit.
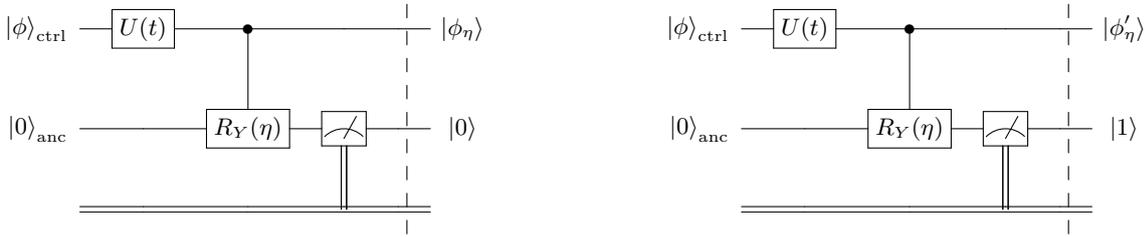
\begin{figure}[h]
     \centering
     \begin{subfigure}[b]{0.496\textwidth}
        \begin{align*}
        \Qcircuit @C=1.3em @R=2.5em {
        \lstick{\ket{\phi}_{\text{ctrl}}} & \gate{U(t)} & \ctrl{1} & \qw & \qw \barrier{2} & \qw & \ket{\phi_{\eta}} \\
        \lstick{\ket{0}_{\text{anc}}} & \qw & \gate{R_Y(\eta)}  & \meter \cwx[1] & \qw & \qw & \ket{0} \\ 
         & \cw & \cw & \cw \cwx[-1] & \cw & \cw \\
        }
        \end{align*}
        \caption{state after the "NO" event: $\ket{\Psi(\eta)} = \ket{\phi_{\eta}}\otimes\ket{0}$}
     \end{subfigure}
     \hfill
     \begin{subfigure}[b]{0.496\textwidth}
          \begin{align*}
        \Qcircuit @C=1.3em @R=2.5em {
        \lstick{\ket{\phi}_{\text{ctrl}}} & \gate{U(t)} & \ctrl{1} & \qw & \qw \barrier{2} & \qw & \ket{\phi'_{\eta}} \\
        \lstick{\ket{0}_{\text{anc}}} & \qw & \gate{R_Y(\eta)}  & \meter \cwx[1] & \qw & \qw & \ket{1} \\ 
         & \cw & \cw & \cw \cwx[-1] & \cw & \cw \\
        }
        \end{align*}
        \caption{state after the "YES" event $\ket{\Psi(\eta)} = \ket{\phi'_{\eta}}\otimes\ket{1}$}
     \end{subfigure}
    \caption{State before and after it is measured weakly}
\end{figure}
\begin{proposition}\label{prop:state_after_measurement}
    Let $\ket{\psi_t} =  
         \begin{pmatrix}
           \alpha_t \\
           \beta_t
         \end{pmatrix}
         $
    be some arbitrary state on the control qubit after $U(t)$ has been applied. Then, the state after the null-type weak measurement on the ancilla qubit is a separable, sub-normalized, two-qubit quantum state of the form
    \begin{equation}
  \ket{\Psi(\eta)} = \begin{cases}
        \ket{\phi'_{\eta}}\otimes\ket{1}, \, \text{if the particle has been detected, i.e. "YES"} \\
         \ket{\phi_{\eta}}\otimes\ket{0}, \, \text{if the particle has not been detected, i.e. "NO"}.
        \end{cases}
    \end{equation}
    Here, $\ket{\phi'_{\eta}}:= 
        \begin{pmatrix}
           0 \\
           \beta_t b(\eta)
         \end{pmatrix} $
    and 
    $\ket{\phi_{\eta}}:= 
        \begin{pmatrix}
           \alpha_t \\
           \beta_t a(\eta)
         \end{pmatrix} $, 
         where $a(\eta):=\sqrt{1-\eta}$ and $b(\eta):= \sqrt{\eta}$ originate from the weak measurement, for some $\eta \in (0,1].$ 
\end{proposition}
\begin{proof}
    The weak measurement consists of two steps. First, we apply the controlled $R_y$ rotation on the ancilla qubit and second, we measure the ancilla in the computational basis. 
    The controlled $R_y$ rotation applied to the state $\ket{\psi_t}\otimes\ket{0}_{\text{anc}}$ gives:
    \begin{align*}
        V(\eta)\left(\ket{\psi_t}\otimes \ket{0}\right) &= (\ketbra{0}{0}\otimes\mathds{1} + \ketbra{1}{1}\otimes U(\eta))(\ket{\psi_t}\otimes \ket{0}) = \alpha_t\ket{0}\ket{0}+\beta_t\ket{1}(a(\eta)\ket{0}+b(\eta)\ket{1})
    \end{align*}
    We define the state after this rotation by $\ket{\varphi_t}:=\alpha_t\ket{0}\ket{0}+\beta_t\ket{1}(a(\eta)\ket{0}+b(\eta)\ket{1}).$
    Now, the measurement is applied to the ancilla qubit. Let $Q = \ketbra{0}{0}$ and $P=\ketbra{1}{1}$, so that we get
    \begin{align*}
        P \ket{\varphi_t} = \beta_t b(\eta)\ket{1}\otimes\ket{1}, \qquad 
         Q \ket{\varphi_t} = (\alpha_t\ket{0} + \beta_t a(\eta)\ket{1})\otimes\ket{0}. 
    \end{align*}
    Note that the projector $P$ corresponds to a "YES" and $Q$ to a "NO" measurement. By defining $\ket{\phi'_{\eta}}:=\beta_t b(\eta)\ket{1}$ and $\ket{\phi_{\eta}}:=\alpha_t\ket{0} + \beta_t a(\eta)\ket{1}$, we have the assertion.  
\end{proof}
The post-measurement state is obtained by normalizing $P\ket{\varphi_t},$ respectively $Q\ket{\varphi_t}$. Proposition \ref{prop:state_after_measurement} shows that there is no additional reset operation needed in the protocol. Mathematically, the ancilla remains in state $\ket{0}$ after every unsuccessful measurement. 
\subsection{Weak Measurements: Subnormalized States and Subspace Recurrence}\label{subseq:appendix_weak_measurements}
\subsubsection{Shear Mapping of the Projection Operators}\label{subseq:appendix_shear_mapping}
By Proposition \ref{prop:state_after_measurement}, we can now say how the orthogonal projections are stretched to those linear mappings that represent the weak measurement. 
The transformation matrix from (\ref{eq:shear_projectors}),
$
    A_{\eta} := \begin{pmatrix}
        \sqrt{\eta} & 0 \\
        \sqrt{1-\eta} & 1
    \end{pmatrix}
$
is singular if and only if $\eta = 0$. For $\eta = 1$, $A_1 = \mathds{1}_2$ and all other values $0<\eta<1$ induce a shear mapping of the observables. Since the weak measurement is unitarily equivalent to a phase-flip channel, a picture of how the above mapping acts on qubits is seen e.g. in Chapter 8 from Nielsen and Chuang \cite{nielsen2010quantum}.

It is recaptured that the encoding is invariant under any unitary basis transformation.
Let $B$ denote a unitary complex $2\times2$ matrix, which transforms the computational basis, i.e. 
\begin{equation*}
    P \mapsto B^{\dagger} P B =: P', \qquad Q\mapsto B^{\dagger} Q B =: Q'. 
\end{equation*}
Then, using Eq. (\ref{eq:shear_projectors}),
\begin{align}\label{eq:shear_projectors_basis_transform}
\begin{split}
    P_{\eta} &\mapsto B^{\dagger} P_{\eta} B = \sqrt{\eta} B^{\dagger} P B = \sqrt{\eta} P' \equiv P'_{\eta}, \\
    Q_{\eta} &\mapsto B^{\dagger} Q_{\eta} B = B^{\dagger} \left(Q+\sqrt{1-\eta} P\right) B = B^{\dagger} Q B + \sqrt{1-\eta} B^{\dagger} P B = Q' + \sqrt{1-\eta} P' \equiv Q'_{\eta}.
\end{split}
\end{align}
An important application is, when the diagonal basis $\{\ket{+}, \ket{-}\}$ is chosen.  In this case, $B = H$, where $H$ is the \textit{Hadamard gate} so that
$\ket{0} \mapsto H \ket{0} =: \ket{+}$ and $\ket{1} \mapsto H \ket{1} =: \ket{-}$.
Explicitly,
$H = \frac{1}{\sqrt{2}} \begin{pmatrix}
    1 & 1 \\
    1 & -1
\end{pmatrix}$, yielding with Eq. (\ref{eq:shear_projectors_basis_transform})
\begin{align*}
    P'_{\eta} = \frac{1}{2} \begin{pmatrix}
        \sqrt{\eta} & -\sqrt{\eta} \\
        -\sqrt{\eta} & \sqrt{\eta} \\
    \end{pmatrix},
    \qquad
     Q'_{\eta} = \frac{1}{2} \begin{pmatrix}
        1+ \sqrt{1-\eta} & 1-\sqrt{1-\eta} \\
        1-\sqrt{1-\eta} & 1+\sqrt{1-\eta} \\
    \end{pmatrix}.
\end{align*}
\begin{definition}\label{def:weak_mappings}
Consider $S^{\sqrt{\eta}}_V=\{\ket{\psi}\in V: \norm{\ket{\psi}} = \sqrt{\eta}, \, \forall \,\eta \in (0,1] \}.$ The orthogonal projection $P:\mathcal{H}\to V$ gives a linear (but not projection-valued) mapping $P_{\eta}:\mathcal{H}\to V, P_{\eta}:=\sqrt{\eta} P$. Similarly, the mapping $Q_{\eta}$ is given by (\ref{eq:shear_projectors}), i.e. $Q_{\eta}=Q+\sqrt{1-\eta}P$.
\end{definition}
\begin{proposition}\label{prop:yes_no_condition}
Consider $S^{\sqrt{1-\eta}}_V= \{\ket{\psi}\in V: \norm{\ket{\psi}} = \sqrt{1 - \eta}\}$ and $S^{\sqrt{\eta}}_V= \{\ket{\psi}\in V: \norm{\ket{\psi}} = \sqrt{\eta}\}$. 
    \begin{enumerate}
        \item[(i)] Let $\ket{\psi}\in S_V$. Then, $P_{\eta}\ket{\psi}\in S^{\sqrt{\eta}}_V$ and $Q_{\eta} \ket{\psi} \in  S^{\sqrt{1-\eta}}_V$.
        \item[(ii)] Let $\ket{\psi}\in S_{V^{\perp}}$. Then, $P_{\eta}\ket{\psi} = 0$ and $Q_{\eta}\ket{\psi}\in S_{V^{\perp}}$.
    \end{enumerate}
\end{proposition}
\begin{proof}
   The statements follow immediately by definition of $P_{\eta}$ and $Q_{\eta}$: 
   \begin{itemize}
       \item[(i)] Let $\ket{\psi}\in S_V$. Then $P_{\eta}\ket{\psi} = \sqrt{\eta}\ket{\psi}\in S_V^{\sqrt{\eta}}.$ Moreover, $Q_{\eta}\ket{\psi} = \sqrt{1-\eta}\ket{\psi}\in S_V^{\sqrt{1-\eta}}.$ 
    \item[(ii)] Let $\ket{\psi}\in S_{V^{\perp}}$. Then $P_{\eta}\ket{\psi} = \sqrt{\eta}P\ket{\psi} = 0.$ Moreover, $Q_{\eta}\ket{\psi} = Q \ket{\psi} = \ket{\psi}\in S_{V^{\perp}}.$ 
   \end{itemize}
\end{proof}
\subsubsection{Weak Recurrence}
For shorter notation, we sometimes abbreviate $U:=U(t)$ and $\Tilde{U}_{\eta}:=\Tilde{U}_{\eta}(t)$. It is mentioned that any unitary operator $U$ depends continuously on time $t\in\mathbb{R},$ i.e. $U\equiv U_t = U(t) = e^{-iHt}$.
\begin{lemma}\label{lem:return_prob_and_time}
Let $\ket{\psi}\in S_V$ be an initial state and $U(t)$ be a unitary operator. Then, after measuring $N\in\mathbb{N}$ times, the $\eta-$dependent return probability to $V$, denoted by $R_N(\eta, t)$ and the expected first return time, denoted by $\tau_N(\eta, t)$ are
\begin{equation*}\label{eq:return_prob_eta_appendix}
    R_N(\eta, t) = 1 - \norm{\Tilde{U}_{\eta}(t)^N\ket{\psi}}^2, 
\end{equation*}
\begin{equation*}\label{eq:return_time_eta_appendix}
    \tau_N(\eta, t)=\frac{\sum_{k=0}^{N-1}\norm{\Tilde{U}_{\eta}(t)^k\ket{\psi}}^2 - N \norm{\Tilde{U}_{\eta}(t)^N\ket{\psi}}^2}{1-\norm{\Tilde{U}_{\eta}(t)^N\ket{\psi}}^2}.
\end{equation*}
\end{lemma}
\begin{proof}
We follow the procedure of \cite{bourgain2014quantum} and compute an anologous decomposition of the form $\Hat{a}_n + \Tilde{U}^n P = U \Tilde{U}^{n-1} P.$ 
Naturally, the "pseudo-orthogonality" of (\ref{eq:pseudo_orthogonal}) is inherited in such a decomposition for weak measurements. 
For $\ket{\psi}\in S_V$, Proposition \ref{prop:yes_no_condition} yields  $P_{\eta}\ket{\psi}\in S^{\sqrt{\eta}}_V$.  
The $\eta$-dependent quantities become (compare with Def. \ref{def:weak_mappings})
\begin{align*}
    P \mapsto P_{\eta}, \qquad 
    \Tilde{U} \mapsto \Tilde{U}_{\eta}:= Q_{\eta}U, \qquad
    \Hat{a}_n \mapsto \Hat{a}^{\eta}_{n} := P_{\eta} U \Tilde{U}^{n-1}_{\eta} P.
\end{align*}
Recall that the projector on the right hand side controls that $\ket{\psi}\in ran \,P.$ This is not a measurement after the zero'th step, so that there is no $\eta$-term arising here. The modified decomposition reads 
\begin{equation}\label{eq:pseudo_orthogoanl_decomp}
    \Hat{a}^{\eta}_n + \Tilde{U}^{n}_{\eta} P = (\mathds{1} + f(\eta) P) U \Tilde{U}^{n-1}_{\eta} P.
\end{equation}
Details are presented in Appendix \ref{proof:lem_return_prob_and_time}. In order to transfer the construction from \cite{bourgain2014quantum}, we have to show that the telescoping series for the return probability is preserved in the case of measuring weakly. 
Indeed, one computes (see Appendix \ref{proof:lem_return_prob_and_time})
\begin{equation}\label{eq:single_return_prob}
    p^{\eta}_n(t):=\norm{\Hat{a}^{\eta}_n(t)\ket{\psi}}^2 = \norm{\Tilde{U}_{\eta}(t)^{n-1}\ket{\psi}}^2 - \norm{\Tilde{U}_{\eta}(t)^n\ket{\psi}}^2.
\end{equation}
By definition of the overall $V-$return probability \cite{bourgain2014quantum} and the above telescoping sum,
\begin{equation*}
     R_N(\eta, t):= \sum_{n=1}^{N} p^{\eta}_n(t) =  \sum_{n=1}^{N}  \norm{\Hat{a}^{\eta}_n(t)\ket{\psi}}^2  = 1 - \norm{\Tilde{U}_{\eta}(t)^N\ket{\psi}}^2,
\end{equation*}
which shows (\ref{eq:return_prob_eta}).
We insert the above expression for the expected first V-return time \cite{grunbaum2013recurrence}, i.e. 
$\tau_N = \frac{\sum_{n=1}^{N} n\, p_n}{R_N},$ so that we have for $\tau_N(\eta, t)$:
\begin{equation*}
    \tau_N(\eta, t) = \frac{\sum_{n=1}^{N} n \, \norm{\Hat{a}^{\eta}_n(t)\ket{\psi}}^2} {\sum_{n=1}^{N} \norm{\Hat{a}^{\eta}_n(t) \ket{\psi}}^2} = \frac{\sum_{k=0}^{N-1}\norm{\Tilde{U}_{\eta}(t)^k\ket{\psi}}^2 - N \norm{\Tilde{U}_{\eta}(t)^N\ket{\psi}}^2}{1-\norm{\Tilde{U}_{\eta}(t)^N\ket{\psi}}^2}.
\end{equation*}
\end{proof}
\begin{corollary}\label{cor:total_probability_weak}
    Let $\ket{\psi}\in S_V$ be $V-$recurrent, i.e. $R_{\infty}(t)=1.$ Then, $\ket{\psi}$ is also $V-$recurrent in the weak measurement protocol, i.e. $R_{\infty}(\eta, t) = 1$, for all $\eta \in (0,1]$.
\end{corollary}
\begin{proof}
For $\eta=1$, the statement is trivial. Let $\eta \in (0,1).$ Then, we compute:
\begin{align*}
    R_{\infty}(\eta, t) = \lim_{N\to\infty}R_N(\eta, t) &=  \lim_{N\to\infty}\left( 1 - \norm{\Tilde{U}_{\eta}(t)^N\ket{\psi}}^2  \right) =  1 - \underbrace{\lim_{N\to\infty}\norm{\Tilde{U}_{\eta}(t)^N\ket{\psi}}^2}_{\overset{\eta>0}{\longrightarrow} 0 } =
    1.
\end{align*}
The fact that $\lim_{N\to\infty}\norm{\Tilde{U}_{\eta}(t)^N\ket{\psi}}^2 = 0$ holds true, is seen as follows:
For $\eta\in (0,1]$, we have for any unit vectors $\ket{\psi}$ and $\ket{\psi'}:=U\ket{\psi}$, 
\begin{equation}\label{eq:norm_relation}
    0 \leq \norm{\Tilde{U}_{\eta}\ket{\psi}}^2 = \bra\psi \Tilde{U}_\eta^\dagger \Tilde{U}^{}_\eta\ket{\psi}
    =\bra{\psi'}\mathds{1}-\eta P\ket{\psi'}
    =1-\eta\bra{\psi'}P P\ket{\psi'} \leq 1
    ,
\end{equation}
since $P$ is a projector (i.e., $P=P^2$).
Moreover, Eq. (\ref{eq:norm_relation}) yields
$\bra\psi \Tilde{U}_\eta^\dagger \Tilde{U}^{}_\eta\ket{\psi}=1,$ if and only if $P\ket{\psi'}=PU\ket{\psi}=0$.
Focusing on this case, we define $\ket{\Tilde{\psi}}:=\Tilde{U}_\eta^{N-1}\ket{\psi}$ and assume that $\norm{\Tilde{U}_{\eta}\ket{\Tilde{\psi}}}^2 = \bra{\Tilde{\psi}} \Tilde{U}_\eta^\dagger \Tilde{U}^{}_\eta\ket{\Tilde{\psi}}=1,$ for some initial state $\ket{\psi}\in S_V$. Then, 
\begin{align*}
1 = \bra{\psi}\Tilde{U}_{\eta}^{\dagger N} \Tilde{U}_{\eta}^{N}\ket{\psi} = \bra{\psi}\Tilde{U}_{\eta}^{\dagger N-1} \Tilde{U}_{\eta}^{\dagger} \Tilde{U}_{\eta}  \Tilde{U}_{\eta}^{N-1}\ket{\psi} = \bra{\Tilde{\psi}} \Tilde{U}_{\eta}^{\dagger} \Tilde{U}_{\eta}\ket{\Tilde{\psi}} = \bra{\Tilde{\psi}} U^{\dagger} (\mathds{1} -\eta P)U\ket{\Tilde{\psi}} = 1-\eta \bra{\Tilde{\psi}} U^{\dagger} P U\ket{\Tilde{\psi}}
\end{align*}
implies that $PU\ket{\Tilde{\psi}}=0$.
Thus, when $\ket{\Tilde{\psi}}$ is a unit vector and $\ket{\psi}\in S_V$ is an initial state,
\begin{equation*}
|\bra{\psi}U\ket{\Tilde{\psi}}|^2 = |\bra{\psi} P U\ket{\Tilde{\psi}}|^2 = |\bra{\psi}PU \Tilde{U}_\eta^{N-1}\ket{\psi}|^2 = 0,
\end{equation*}
which implies that the return probability $ |\bra{\psi} \Hat{a}_N^{\eta}\ket{\psi}|^2 := |\bra{\psi}P_{\eta}U \Tilde{U}_\eta^{N-1}\ket{\psi}|^2 = \eta|\bra{\psi}PU \Tilde{U}_\eta^{N-1}\ket{\psi}|^2 = 0.$ Hence, we conclude that $\norm{\ket{\Tilde{\psi}}} = \norm{\Tilde{U}_\eta^{N-1}\ket{\psi}} <1$ for all initial states $\ket{\psi}\in S_V$ and $N\geq 2,$ i.e. for $M:=N-1$, $\lim_{M\to\infty}\norm{\Tilde{U}_{\eta}(t)^M\ket{\psi}}^2 = 0.$

Note that, in the case of $\eta =0$ (full decoupling), $ \lim_{N\to\infty}\norm{\Tilde{U}_{\eta}(t)^N\ket{\psi}}^2 = \lim_{N\to\infty}\norm{U(t)^N\ket{\psi}}^2 = 1$, i.e. $R_{\infty}(0, t) = 0,$ coinciding with the intuition that no state is recurrent, when nothing is being measured.   
\end{proof}
\begin{corollary}\label{cor:return_prob_estimates}
Let $R_N(\eta, t)$ denote the $V-$return probability of an initial state $\ket{\psi}\in S_V$ after $N$ weak measurements with coupling strength $\eta\in (0,1]$. Then, 
$R_N(\eta, t) \leq R_N(1,t) \equiv R_N(t)$. Clearly, also: $R_{\infty}(\eta,t) \leq R_{\infty}(1,t) \equiv R_{\infty}(t)$, where $R_{\infty}(\eta,t):= \lim_{N\to\infty} R_N(\eta, t).$  
\end{corollary}
\begin{proof}
As usual, $\eta \in (0,1].$ We distinguish between two cases. First, $\eta = 1$ and second, $0<\eta<1.$ In the first case, we always have to recover the projection-valued quantities from \cite{bourgain2014quantum}. Due to $\Tilde{U}_{\eta}:= Q_{\eta}U = (Q+\sqrt{1-\eta}P)U$, we have $\Tilde{U}_1 = QU \equiv \Tilde{U}.$ Thus, $R_N(t,1) =  1-\norm{\Tilde{U}_{1}(t)^N\ket{\psi}}^2 = 1-\norm{\Tilde{U}(t)^N\ket{\psi}}^2 \equiv R_N(t).$ We take the limit $N\to\infty$ and have $ R_{\infty}(t) = 1- \lim_{N\to\infty}\norm{\Tilde{U}(t)^N\ket{\psi}}^2.$ \\
Now, let $0<\eta<1.$ For two positive (diagonal) operators $\Lambda_1$ and $\Lambda_2$,  $N\in\mathbb{N},$ it holds true that
\begin{equation*}
    \norm{\left((\Lambda_1 + \Lambda_2)U\right)^N\ket{\psi}} \geq \norm{\left(\Lambda_1U\right)^N\ket{\psi}}, \qquad \ket{\psi}\in S_V
\end{equation*}
\begin{align*}
  \Longrightarrow \quad  \norm{\Tilde{U}^N_{\eta}\ket{\psi}}^2 = \norm{\left(\left(Q + \sqrt{1-\eta}P\right)U\right)^N\ket{\psi}}^2 \geq  \norm{(QU)^N\ket{\psi}}^2 = \norm{\Tilde{U}^N\ket{\psi}}^2.
\end{align*}
Consequently, 
\begin{align*}
R_N(\eta, t) = 1 - \norm{\Tilde{U}_{\eta}(t)^N\ket{\psi}}^2 \leq 1 - \norm{\Tilde{U}(t)^N\ket{\psi}}^2 \equiv R_N(t).
\end{align*}
The monotonicity of taking the limit proves the assertion.
\end{proof}
%
%
%
%
%
\begin{corollary}\label{cor:return_time_estimates}
Let $\tau_N(\eta,t)$ denote the expected $V-$ first return time of an initial state $\ket{\psi}\in S_V$ after $N$ weak measurements with coupling strength $\eta\in (0,1]$. Then, 
$\tau_N(\eta,t) \geq \tau_N(1,t) \equiv \tau_N(t)$. Clearly, also: $\tau_{\infty}(\eta, t) \geq \tau_{\infty}(1, t) \equiv \tau_{\infty}(t)$, where $\tau_{\infty}(\eta, t):= \lim_{N\to\infty} \tau_N(\eta, t).$  
\end{corollary}
\begin{proof}
We apply Corollary \ref{cor:return_prob_estimates} and note that for all $n\geq 1,$
\begin{align*}
    \norm{\Tilde{U}^{n}_{1}\ket{\psi}}^2 = \norm{\Tilde{U}^{n}\ket{\psi}}^2 &= \norm{(QU)^{n}\ket{\psi}}^2 \leq \norm{\left( (Q+\sqrt{1-\eta}P)U\right)^{n}\ket{\psi}}^2 \leq 1, \\
    \Longrightarrow 1 - \norm{\Tilde{U}^n\ket{\psi}}^2 &\geq  1 - \norm{\Tilde{U}^n_{\eta}\ket{\psi}}^2. 
\end{align*}
From this, and $\norm{\Tilde{U}_{\eta}^k\ket{\psi}}\geq \norm{\Tilde{U}^k\ket{\psi}}$ for all $k \geq 1$ (cf. Cor. \ref{cor:return_prob_estimates}), it follows that
\begin{align*}
    \tau_N(\eta, t) &= \frac{\sum_{k=0}^{N-1}\norm{\Tilde{U}_{\eta}(t)^k\ket{\psi}}^2 - N \norm{\Tilde{U}_{\eta}(t)^N\ket{\psi}}^2}{1-\norm{\Tilde{U}_{\eta}(t)^N\ket{\psi}}^2} \geq \frac{\sum_{k=0}^{N-1}\norm{\Tilde{U}(t)^k\ket{\psi}}^2 - N \norm{\Tilde{U}_{\eta}(t)^N\ket{\psi}}^2}{1-\norm{\Tilde{U}_{\eta}(t)^N\ket{\psi}}^2} \\
    &\geq 
    \frac{\sum_{k=0}^{N-1}\norm{\Tilde{U}(t)^k\ket{\psi}}^2 - N \norm{\Tilde{U}(t)^N\ket{\psi}}^2}{1-\norm{\Tilde{U}_{\eta}(t)^N\ket{\psi}}^2} 
    \geq \frac{\sum_{k=0}^{N-1}\norm{\Tilde{U}(t)^k\ket{\psi}}^2 - N \norm{\Tilde{U}(t)^N\ket{\psi}}^2}{1-\norm{\Tilde{U}(t)^N\ket{\psi}}^2} \overset{\text{def.}}{=} \tau_N(t).
\end{align*}
The monotonicity of taking the limit proves the assertion.
\end{proof}
\begin{lemma}\label{lem:one_over_eta}
Let $\ket{\psi}\in S_V$. For $t=0$ and for all $\eta\in (0,1]$, the expected return time $\tau_{N}(\eta, t=0)$ is given by:
\begin{equation*}\label{eq:one_over_eta_law_appendix}
    \tau_N(\eta,0)=\frac{1}{\eta} - \frac{N (1-\eta)^N}{1-(1-\eta)^N}.
\end{equation*}
Especially, $\tau_{\infty}(\eta,0)=\frac{1}{\eta}$. The rate, of which $\tau_N(\eta,0)$ converges point wise to $ \frac{1}{\eta}$ is $c(\eta)= 1-\eta.$
\end{lemma} 
\begin{proof}
For $t=0$, $U(0)=\mathds{1}$. As mentioned in the text, if $t=0$, clearly $ \tau_N(1, 0)= 1.$ So let $\eta \in (0,1)$.
Since there is no evolution at all,  $\Tilde{U}_{\eta}\equiv \Tilde{\mathds{1}}_{\eta}  =\left( (Q+\sqrt{1-\eta}P)\mathds{1}\right)$, which implies $\Tilde{\mathds{1}}_{\eta}^k =Q+(\sqrt{1-\eta})^kP$ due to $P, Q$ being projection operators. Then, for all $\ket{\psi}\in S_V$, $\norm{\Tilde{\mathds{1}}_{\eta}^k\ket{\psi}}^2 = (1-\eta)^k.$ Inserting this in Eq. (\ref{eq:weak_first_return_time_N_telescoping}), it follows
\begin{equation*}
    \tau_{N}(\eta,0)= \frac{\sum_{k=0}^{N-1}(1-\eta)^k - N (1-\eta)^N}{1-(1-\eta)^N}. 
\end{equation*}
Since $|1-\eta|<1$, we can simply use the geometric sum $\sum_{k=0}^{N-1}q^k = \frac{1-q^N}{1-q}$ ($|q|<1$), from which we get
\begin{align*}
    \tau_{N}(\eta,t=0) &= \frac{\frac{1-(1-\eta)^N}{1-(1-\eta)} -N (1-\eta)^N}{1-(1-\eta)^N} 
    = \frac{1}{\eta} - \frac{N (1-\eta)^N}{1-(1-\eta)^N}.
\end{align*}
Taking the limit $N\to\infty$ yields $\tau_{\infty}(\eta,0)=\frac{1}{\eta}.$
%

The rate of convergence is defined by  $c(\eta):= \lim_{N\to\infty}\frac{\vert \tau_{N+1}(\eta,0) - \tau_{\infty}(\eta,0)\vert}{\vert \tau_{N}(\eta,0) - \tau_{\infty}(\eta,0)\vert}$. Simple computation yields
\begin{align*}
    c(\eta) &= \lim_{N\to\infty}\frac{\left\vert \frac{(N+1)(1-\eta)^{N+1}}{1-(1-\eta)^{N+1}}\right\vert}{\left\vert  \frac{N(1-\eta)^{N}}{1-(1-\eta)^{N}}\right\vert} 
    = \lim_{N\to\infty} \frac{(N+1)(1-\eta)^{N+1}(1-(1-\eta)^N)}{N(1-\eta)^N \left(1-(1-\eta)^{N+1}\right)} \\
    &= \lim_{N\to\infty} \frac{\left(1-\eta+\frac{1-\eta}{N}\right)\left(1-(1-\eta)^N \right)}{1-(1-\eta)^{N+1}} 
    = 1-\eta.
\end{align*}
\end{proof}
\section{Details of Proofs}\label{appendix:proofs}
\subsection{Lemma \ref{lem:return_prob_and_time}}\label{proof:lem_return_prob_and_time}
In order to show Eq. (\ref{eq:pseudo_orthogoanl_decomp}), we start by computing the left side of (\ref{eq:pseudo_orthogoanl_decomp}) and conclude the right. 
\begin{align*}
     \Hat{a}^{\eta}_n + \Tilde{U}^{n}_{\eta} P &=  P_{\eta} U \Tilde{U}^{n-1}_{\eta} P +  \Tilde{U}^{n}_{\eta} P
     = P_{\eta} U \Tilde{U}^{n-1}_{\eta} P +  \Tilde{U}_{\eta} \Tilde{U}^{n-1}_{\eta} P
     = \left( P_{\eta} U +  \Tilde{U}_{\eta} \right) \Tilde{U}^{n-1}_{\eta} P \\
      &= \left( P_{\eta} U +  Q_{\eta} U \right) \Tilde{U}^{n-1}_{\eta} P 
        = \left( P_{\eta}  +  Q_{\eta} \right) U \Tilde{U}^{n-1}_{\eta} P
          = \left( \sqrt{\eta} P  +  (\mathds{1} - P ) + \sqrt{1-\eta} P \right) U \Tilde{U}^{n-1}_{\eta} P \\
          &= \left(\mathds{1} + f(\eta) P \right) U \Tilde{U}^{n-1}_{\eta} P,
\end{align*}
where 
\begin{equation}
f:(0,1]\to [0,\sqrt{2}-1], \qquad f(\eta):=\sqrt{\eta}+\sqrt{1-\eta}-1.    
\end{equation}
Moreover, equation (\ref{eq:single_return_prob}) is shown.
From the above pseudo-orthogonal decomposition, one writes
\begin{align*}
    \Hat{a}^{\eta}_n = U \Tilde{U}^{n-1}_{\eta} P - \Tilde{U}^{n}_{\eta} P + f(\eta) P U \Tilde{U}^{n-1}_{\eta} P.
\end{align*}
For $\ket{\psi}\in S_V,$ we have by construction $P\ket{\psi} = \ket{\psi}.$ Thus, 
\begin{align*}
    \Hat{a}^{\eta}_n\ket{\psi} = \underbrace{U \Tilde{U}^{n-1}_{\eta}\ket{\psi} - \Tilde{U}^{n}_{\eta}\ket{\psi}}_{=:a} + \underbrace{f(\eta) P U \Tilde{U}^{n-1}_{\eta}\ket{\psi}}_{=:b},
\end{align*}
so that $ \norm{\Hat{a}^{\eta}_n\ket{\psi}}^2$ is of the form  
\begin{align*}
    \norm{\Hat{a}^{\eta}_n\ket{\psi}}^2 = \norm{a+b}^2 = \norm{a}^2 + \norm{b}^2 + 2\Re \, \langle a ,b \rangle.
\end{align*}
Here, we have defined $a:=(U \Tilde{U}^{n-1}_{\eta} - \Tilde{U}^{n}_{\eta})\ket{\psi}$ and $b:= f(\eta) P U \Tilde{U}^{n-1}_{\eta}\ket{\psi}$. Hence, we compute separately the subsequent three expressions: $\norm{a}^2$, $\norm{b}^2$, $2\Re \,\langle a ,b \rangle$. \\
The results are:
\begin{enumerate}
    \item $ \norm{a}^2 = - \norm{\Tilde{U}^{n-1}_{\eta}\ket{\psi}}^2 + \norm{\Tilde{U}^{n}_{\eta}\ket{\psi}}^2 + 2 \, \frac{g(\eta)}{\eta}\norm{\Hat{a}^{\eta}_n\ket{\psi}}^2, $
    \item $ \norm{b}^2 = \frac{f^2(\eta)}{\eta} \norm{\Hat{a}^{\eta}_n\ket{\psi}}^2, $
    \item $ 2\Re \,\langle a ,b \rangle = \frac{2 g(\eta) f(\eta)}{\eta}  \norm{\Hat{a}^{\eta}_n\ket{\psi}}^2. $
\end{enumerate}
Details are presented in Section \ref{subsec:few_details}. 
For simplicity, we define $h(\eta):=\frac{1}{\eta}\left(2 g(\eta) + f^2(\eta)+ 2 g(\eta) f(\eta) \right) = 2.$
Then, we find by combining all three expressions:
\begin{align*} 
\norm{\Hat{a}^{\eta}_n\ket{\psi}}^2 &= - \norm{\Tilde{U}^{n-1}_{\eta}\ket{\psi}}^2 + \norm{\Tilde{U}^{n}_{\eta}\ket{\psi}}^2 + 2 \, \frac{g(\eta)}{\eta}\norm{\Hat{a}^{\eta}_n\ket{\psi}}^2 + \frac{f^2(\eta)}{\eta} \norm{\Hat{a}^{\eta}_n\ket{\psi}}^2 + \frac{2 g(\eta) f(\eta)}{\eta}  \norm{\Hat{a}^{\eta}_n\ket{\psi}}^2 \\
&=  - \norm{\Tilde{U}^{n-1}_{\eta}\ket{\psi}}^2 + \norm{\Tilde{U}^{n}_{\eta}\ket{\psi}}^2 + h(\eta)\, \norm{\Hat{a}^{\eta}_n\ket{\psi}}^2 \\
 &= - \norm{\Tilde{U}^{n-1}_{\eta}\ket{\psi}}^2 + \norm{\Tilde{U}^{n}_{\eta}\ket{\psi}}^2 + 2\, \norm{\Hat{a}^{\eta}_n\ket{\psi}}^2, \\
        \Longleftrightarrow \quad 
     \norm{\Hat{a}^{\eta}_n\ket{\psi}}^2 &= \norm{\Tilde{U}^{n-1}_{\eta}\ket{\psi}}^2 - \norm{\Tilde{U}^{n}_{\eta}\ket{\psi}}^2.
\end{align*}
\subsection{A few Details}\label{subsec:few_details}
\subsubsection{Details to \ref{proof:lem_return_prob_and_time}}
\begin{enumerate}
    \item Such an expression was already computed in \cite{bourgain2014quantum}. It corresponds to the $\eta = 1$ projection-valued measurement. However, we repeat the computation to follow closely what will change in the case of our weak measurement protocol.
    \begin{align*}
        \norm{a}^2 &= \norm{(U \Tilde{U}^{n-1}_{\eta} - \Tilde{U}^{n}_{\eta})\ket{\psi}}^2 
        = \norm{U \Tilde{U}^{n-1}_{\eta}\ket{\psi}}^2 + \norm{\Tilde{U}^{n}_{\eta}\ket{\psi}}^2 - 2 \Re \, \langle U \Tilde{U}^{n-1}_{\eta}\psi, \Tilde{U}^{n}_{\eta}\psi \rangle \\
        &= \norm{U}^2 \norm{\Tilde{U}^{n-1}_{\eta}\ket{\psi}}^2 + \norm{\Tilde{U}^{n}_{\eta}\ket{\psi}}^2 - 2 \Re \, \langle U \Tilde{U}^{n-1}_{\eta}\psi, \Tilde{U}_{\eta} \Tilde{U}^{n-1}_{\eta}\psi \rangle \\
          &= \norm{\Tilde{U}^{n-1}_{\eta}\ket{\psi}}^2 + \norm{\Tilde{U}^{n}_{\eta}\ket{\psi}}^2 - 2 \Re \, \langle U\psi', \Tilde{U}_{\eta} \psi' \rangle,
    \end{align*}
    where we have defined $\ket{\psi '}:= \Tilde{U}^{n-1}_{\eta}\ket{\psi}.$ \\
    We now compute $\Re \, \langle U\psi', \Tilde{U}_{\eta} \psi' \rangle$. Recall that $Q_{\eta} = Q + \sqrt{1-\eta}P = \mathds{1} - P + \sqrt{1-\eta}P  = \mathds{1} -(1-\sqrt{1-\eta})P=: \left( \mathds{1} - g(\eta) \right) P$ for $g: (0,1] \to (0,1], \, g(\eta):=1-\sqrt{1-\eta}.$
    \begin{align*}
         \langle U\psi', \Tilde{U}_{\eta} \psi' \rangle &= \langle U\psi', (Q_{\eta}U) \psi' \rangle = \langle \psi', U^{*}(Q_{\eta}U) \psi' \rangle \\
         &= \langle \psi', U^{*}(\mathds{1} - g(\eta) P)U \psi' \rangle 
         = \langle \psi', (U^{*}U - U^{*}g(\eta) PU) \psi' \rangle 
         = \langle \psi', (\mathds{1} - U^{*}g(\eta) PU) \psi' \rangle \\
         &= \langle \psi', \psi' \rangle - g(\eta) \langle \psi', U^{*} PU \psi' \rangle 
         = \langle \psi', \psi' \rangle - g(\eta) \langle \psi', U^{*} P^{*}PU \psi' \rangle \\
         &= \langle \psi', \psi' \rangle - g(\eta) \langle \psi', (PU)^{*}PU \psi' \rangle 
         = \langle \psi', \psi' \rangle - g(\eta) \langle PU \psi', PU \psi' \rangle \\
         &= \norm{\ket{\psi'}}^2 - g(\eta) \langle PU \psi', PU \psi' \rangle.
    \end{align*}
Since $\ket{\psi}\in S_V$ and $\eta>0,$ $\ket{\psi} = \frac{1}{\sqrt{\eta}} \left(\sqrt{\eta} P \ket{\psi}\right).$ Using this identity, one computes 
\begin{align*}
     \left\langle PU \psi', PU \psi' \right\rangle &=  \left\langle PU \Tilde{U}^{n-1}_{\eta}\psi , PU \Tilde{U}^{n-1}_{\eta}\psi \right\rangle \\
     &= \left\langle \frac{1}{\sqrt{\eta}} \sqrt{\eta} P U \Tilde{U}^{n-1}_{\eta}  P \psi , \frac{1}{\sqrt{\eta}} \sqrt{\eta} P U \Tilde{U}^{n-1}_{\eta}  P \psi \right\rangle \\
     &=  \frac{1}{\eta} \left\langle P_{\eta} U \Tilde{U}^{n-1}_{\eta} P \psi, P_{\eta} U \Tilde{U}^{n-1}_{\eta} P\psi \right\rangle 
     = \frac{1}{\eta}  \norm{\Hat{a}^{\eta}_n\ket{\psi}}^2. 
\end{align*}
From this, we summarize $  \langle U\psi', \Tilde{U}_{\eta} \psi' \rangle = \norm{\ket{\psi'}}^2 - g(\eta) \langle PU \psi', PU \psi' \rangle = \norm{\ket{\psi'}}^2 - \frac{g(\eta)}{\eta} \norm{\Hat{a}^{\eta}_n\ket{\psi}}^2.$ Thus, 
\begin{equation*}
    \Re \, \langle U\psi', \Tilde{U}_{\eta} \psi' \rangle = \Re \,\left( \norm{\ket{\psi'}}^2 - \frac{g(\eta)}{\eta}\norm{\Hat{a}^{\eta}_n\ket{\psi}}^2 \right) = \norm{\Tilde{U}^{n-1}_{\eta}\ket{\psi}}^2 - \frac{g(\eta)}{\eta}\norm{\Hat{a}^{\eta}_n\ket{\psi}}^2.
\end{equation*}
Combining all together, we finally write
\begin{align*}
    \norm{a}^2 &= \norm{\Tilde{U}^{n-1}_{\eta}\ket{\psi}}^2 + \norm{\Tilde{U}^{n}_{\eta}\ket{\psi}}^2 - 2 \Re \, \langle U\psi', \Tilde{U}_{\eta} \psi' \rangle \\
    &= \norm{\Tilde{U}^{n-1}_{\eta}\ket{\psi}}^2 + \norm{\Tilde{U}^{n}_{\eta}\ket{\psi}}^2 - 2 \left(  \norm{\Tilde{U}^{n-1}_{\eta}\ket{\psi}}^2 - \frac{g(\eta)}{\eta}\norm{\Hat{a}^{\eta}_n\ket{\psi}}^2 \right),  \\
    \Longleftrightarrow \quad 
    \norm{a}^2 &= - \norm{\Tilde{U}^{n-1}_{\eta}\ket{\psi}}^2 + \norm{\Tilde{U}^{n}_{\eta}\ket{\psi}}^2 + 2 \, \frac{g(\eta)}{\eta}\norm{\Hat{a}^{\eta}_n\ket{\psi}}^2.
\end{align*}
From the structure of $\norm{a}^2,$ we already see that the telescoping property is preserved. The expressions $\norm{b}^2$ and $2\Re \,\langle a ,b \rangle$ will erase the pre-factors in $\eta$.
\item Again, since $\ket{\psi}\in S_V$:
\begin{align*}
    b &= f(\eta) P U \Tilde{U}^{n-1}_{\eta}\ket{\psi} = f(\eta) \frac{1}{\sqrt{\eta}} \sqrt{\eta} P U \Tilde{U}^{n-1}_{\eta} P \ket{\psi} 
    = \frac{f(\eta)}{\sqrt{\eta}} \underbrace{P_{\eta} U \Tilde{U}^{n-1}_{\eta} P}_{=\Hat{a}_n^{\eta}}\ket{\psi} 
    = \frac{f(\eta)}{\sqrt{\eta}} \Hat{a}^{\eta}_n\ket{\psi}.
\end{align*}
Hence, we have 
$
\norm{b}^2 = \frac{f^2(\eta)}{\eta} \norm{\Hat{a}^{\eta}_n\ket{\psi}}^2.
$
\item The expression for $b$ was computed in the above step 2. Compute $a$ and combine:
\begin{align*}
    a &= U\Tilde{U}^{n-1}_{\eta} \ket{\psi} - \Tilde{U}^{n}_{\eta}\ket{\psi} = \left(U-\Tilde{U}_{\eta} \right) \Tilde{U}^{n-1}_{\eta}\ket{\psi}  
    = \left(\mathds{1}-Q_{\eta}\right) U \Tilde{U}^{n-1}_{\eta}\ket{\psi}
    = \underbrace{\left(\mathds{1}-Q_{\eta}\right)}_{= g(\eta) P} U \Tilde{U}^{n-1}_{\eta}P\ket{\psi} \\
    &= g(\eta) P  U \Tilde{U}^{n-1}_{\eta}P\ket{\psi} 
    = \frac{g(\eta)}{\sqrt{\eta}} \sqrt{\eta} P U \Tilde{U}^{n-1}_{\eta}P\ket{\psi} 
    = \frac{g(\eta)}{\sqrt{\eta}} \underbrace{P_{\eta} U \Tilde{U}^{n-1}_{\eta}P}_{=\Hat{a}^{\eta}_{n}}\ket{\psi} 
    = \frac{g(\eta)}{\sqrt{\eta}}\Hat{a}^{\eta}_n\ket{\psi}.
\end{align*}
Finally, 
\begin{align*}
    2 \Re\,\langle a, b \rangle &=  2 \Re\,\left\langle \frac{g(\eta)}{\sqrt{\eta}}\Hat{a}^{\eta}_n\psi,  \frac{f(\eta)}{\sqrt{\eta}} \Hat{a}^{\eta}_n\psi \right\rangle  = \frac{2 g(\eta) f(\eta)}{\eta} \Re\,\left\langle \Hat{a}^{\eta}_n\psi, \Hat{a}^{\eta}_n\psi \right\rangle  = \frac{2 g(\eta) f(\eta)}{\eta}  \norm{\Hat{a}^{\eta}_n\ket{\psi}}^2.
\end{align*}
\end{enumerate}
It is mentioned that the telescoping sum in the numerator of (\ref{eq:weak_first_return_time_N_telescoping}) is of the form (compare with Lemma A.3 in \cite{bourgain2014quantum})
\begin{align*}
    \sum_{n=1}^{N} n \cdot ( u_{n-1} - u_n ) &= u_0 - u_1 + 2u_1 - 2u_2 + 3u_2 - 3 u_3  + ... + N(u_{N-1}-u_N) 
    = \sum_{k=0}^{N-1} u_k - N \cdot u_N.
\end{align*}
\subsubsection{Details to the Integral Representation}\label{appendix:details_integral}
%
%
The following statement follows immediately from the findings of Bourgain et al. \cite{bourgain2014quantum} which were subsequently worked out in more generality by Grünbaum and Velazquez in \cite{grunbaum2018generalization}.
Expression (\ref{eq:aharonov_anandan_eta}) is derived in the following. A similar derivation of such a scalar-valued complex return integral was firstly presented in \cite{grunbaum2013recurrence} and can also be found for instance in \cite{yin2019large}.
We explicitly write out Eq. (\ref{eq:aharonov_anandan_eta}) and start by ($\Hat{a}_{n}^{\eta}(t)\equiv \Hat{a}_n^{\eta}$)
\begin{equation*}
\tau_{\infty}(\eta, t) = \sum_{n=1}^\infty n\,\norm{\Hat{a}_n^{\eta}(t)\ket{\psi}}^2. 
\end{equation*}
Note that for $\alpha\in\mathbb{R}$, we can write $n = -i\frac{\partial}{\partial\alpha} e^{in\alpha}\big|_{\alpha = 0}$ and similarly $i n e^{in\theta} = \frac{\partial}{\partial\theta}e^{in\theta}$. Moreover, we will use $\delta_{nm} = \frac{1}{2\pi}\int_{0}^{2\pi} d\theta\, e^{i\theta(n-m)}$ for $n,m\in\mathbb{N}$ and the generating function for the first return operator, i.e. $\Hat{a}(\eta, z,t):=\sum_{k=1}^{\infty}\Hat{a}_{k}^{\eta}(t) z^k,$ which can be written for $|z|=1$ as $\Hat{a}(\eta, e^{i\theta},t):=\sum_{k=1}^{\infty}\Hat{a}_{k}^{\eta}(t) e^{ik\theta},$  for $z:=e^{i\theta}, \theta\in [0,2\pi]$.
Then, 
\begin{align*}
\norm{\Hat{a}_n^{\eta}\ket{\psi}}^2 = \langle\psi,\Hat{a}_n^{\eta \dagger}\Hat{a}_n^{\eta}\psi\rangle = \langle\psi,\delta_{nm}\Hat{a}_n^{\eta \dagger}\Hat{a}_m^{\eta}\psi\rangle = \left\langle\psi, \frac{1}{2\pi}\int_{0}^{2\pi} d\theta\, e^{i\theta(n-m)} \Hat{a}_n^{\eta \dagger}\Hat{a}_m^{\eta}\psi\right\rangle. 
\end{align*}
Combining all, we can write
\begin{align*}
    \tau_{\infty}(\eta, t) &= \sum_{n=1}^\infty n\,\norm{\Hat{a}_n^{\eta}(t)\ket{\psi}}^2 = \sum_{n,m=1}^{\infty}\left(-i \frac{\partial}{\partial \alpha} e^{i n\alpha}\bigg|_{\alpha=0}\right)\left\langle\psi, \frac{1}{2\pi}\int_{0}^{2\pi} d\theta\, e^{i\theta(n-m)} \Hat{a}_n^{\eta \dagger}(t)\Hat{a}_m^{\eta}(t)\psi\right\rangle \\
    &= \sum_{n,m=1}^{\infty}\frac{1}{2\pi i} \int_{0}^{2\pi} d\theta\,\left\langle\psi,e^{-im\theta} \Hat{a}_m^{\eta \dagger}(t) \left(\frac{\partial}{\partial \alpha} e^{i n\alpha}\bigg|_{\alpha=0}\right) e^{in\theta}\Hat{a}_n^{\eta}(t)\psi\right\rangle \\
    &= \sum_{n=1}^{\infty}\frac{1}{2\pi i} \int_{0}^{2\pi} d\theta\,\left\langle\psi,\underbrace{\sum_{m=1}^{\infty}e^{-im\theta} \Hat{a}_m^{\eta \dagger}(t)}_{= \Hat{a}^{\dagger}(\eta, e^{i\theta},t)} \left(\frac{\partial}{\partial \alpha} e^{i n(\alpha+\theta)}\bigg|_{\alpha=0}\right)\Hat{a}_n^{\eta}(t)\psi\right\rangle. \\
    &= \frac{1}{2\pi i} \int_{0}^{2\pi} d\theta\,\left\langle\psi, \Hat{a}^{\dagger}(\eta, e^{i\theta},t) \sum_{n=1}^{\infty}\left(\frac{\partial}{\partial \alpha} e^{i n(\alpha+\theta)}\bigg|_{\alpha=0}\Hat{a}_n^{\eta}(t)\right)\psi\right\rangle.
\end{align*}
Now, we use 
\begin{align*}
    \sum_{n=1}^{\infty} \frac{\partial}{\partial\alpha}\left(e^{in(\alpha+\theta)}\Hat{a}_n^{\eta}(t)\right)\bigg|_{\alpha=0} &= \sum_{n=1}^{\infty} in \,e^{in(\alpha +\theta)}\Hat{a}_n^{\eta}(t)\bigg|_{\alpha=0} =  \sum_{n=1}^{\infty} in \,e^{in \theta}\Hat{a}_n^{\eta}() =  \sum_{n=1}^{\infty} \left(\frac{\partial}{\partial \theta } e^{in \theta}\right) \Hat{a}_n^{\eta}(t) \\
    &= \frac{\partial}{\partial\theta} \sum_{n=1}^{\infty}  e^{in \theta} \Hat{a}_n^{\eta}(t) \overset{\text{def.}}{=} \frac{\partial}{\partial\theta}\Hat{a}(\eta, e^{i\theta},t), 
\end{align*}
which concludes
\begin{equation*}
\tau_{\infty}(\eta, t) = \frac{1}{2\pi i} \int_{0}^{2\pi} d\theta\,\left\langle\psi, \Hat{a}^{\dagger}(\eta, e^{i\theta},t) \frac{\partial}{\partial\theta}\Hat{a}(\eta, e^{i\theta},t)\psi\right\rangle.
\end{equation*}
\section{Analytical Expression of the Two Level System}\label{appendix:symbolic_expression}
Even for an example of minimal dimension, the involved quantities become quickly infeasible to control analytically. 
The involved operators for a complex contour integral are
\begin{align*}
    \Tilde{U}_{\eta}(z,t)&= \frac{U_0(t,z,\eta)}{1 + z^2 \sqrt{1 - \eta} - 
     z (1 + \sqrt{1 - \eta}) \cos t} \\
     \\
    U_0(t,z,\eta)&:=\begin{pmatrix}
    1 - z \cos t  & - iz \sqrt{1-\eta} \sin t  \\
    -i z \sin t & 1 - z \sqrt{1-\eta} \cos t
    \end{pmatrix}
\end{align*}
and
\begin{align*}
 \Hat{a}(\eta, z, t) =    
 \begin{pmatrix}
     \frac{\sqrt{\eta} z \left(\cos t -z\right)}{1 + z^2 \sqrt{1 - \eta} - 
     z (1 + \sqrt{1 - \eta}) \cos t} & 0 \\
      0 & 0
 \end{pmatrix}
\end{align*}
However, evaluating the integral representation of $\tau_{\infty}(\eta, t)$ symbolically does not work, so we take the vectorization approach as described in \cite{grunbaum2018generalization, grunbaum2020quantum}. We define $g(t,x,\eta):= \frac{p(t,x,\eta)}{q(t,x,\eta)}$
with
\begin{align*}
    p(t,x,\eta) =\, & \eta (1 + 
   z (-2 \sqrt{1 - \eta} + 
      z (1 + 2 \sqrt{1 - \eta} + 2 z (-1 + \eta) - \eta)) + \\
      &(1 + z (-2 + z - 2 \sqrt{1 - \eta} + 2 z \sqrt{1 - \eta} - 
         z \eta)) \cos{2 t}) 
\end{align*}
\begin{align*}
    q(t,x,\eta) =\, & 2 (1 + z (-\sqrt{1 - \eta} - z^2 (1 - \eta)^\frac{3}{2} - 
      z \sqrt{1 - \eta} (-2 + \eta) + z^3 (-1 + \eta)^2) \\
   & - z (-1 + z^2 (-1 + \eta)) (-2 + \eta) \cos^2{t} + 
   z (-\sqrt{1 - \eta} + 2 z (1 + \sqrt{1 - \eta} - \eta) \\
   & + z \sqrt{1 - \eta} (z (-1 + \eta) - \eta)) \cos{2 t})
\end{align*}

\end{document}